\documentclass{basic}
\begin{document}
	\title{\bf Nonlinear Permuted Granger Causality}
	\author{Noah D. Gade\footnote{PhD Candidate, Department of Statistics, University of Virginia; Email: ndg5e@virginia.edu}\hspace{.2cm}\\
		and\\
		Jordan Rodu\footnote{Assistant Professor, Department of Statistics, University of Virginia; Email: jsr6q@virginia.edu}}
	\date{\today}
	\maketitle
	\bigskip
	\begin{abstract}
		Granger causal inference is a contentious but widespread method used in fields ranging from economics to neuroscience. The original definition addresses the notion of causality in time series by establishing functional dependence conditional on a specified model. Adaptation of Granger causality to nonlinear data remains challenging, and many methods apply in-sample tests that do not incorporate out-of-sample predictability, leading to concerns of model overfitting. To allow for out-of-sample comparison, a measure of functional connectivity is explicitly defined using permutations of the covariate set. Artificial neural networks serve as featurizers of the data to approximate any arbitrary, nonlinear relationship, and consistent estimation of the variance for each permutation is shown under certain conditions on the featurization process and the model residuals. Performance of the permutation method is compared to penalized variable selection, naive replacement, and omission techniques via simulation, and it is applied to neuronal responses of acoustic stimuli in the auditory cortex of anesthetized rats. Targeted use of the Granger causal framework, when prior knowledge of the causal mechanisms in a dataset are limited, can help to reveal potential predictive relationships between sets of variables that warrant further study.
	\end{abstract}
	
	\noindent%
	\vfill
	\newpage
	\spacingset{1.3}
	
	\section{Introduction}
	Granger causal inference investigates the ability of a time series $\mathbf{x}_t\in\mathbb{R}^{p}$, $t=1,\ldots,T$, to predict future values of a response $\mathbf{y}_t\in\mathbb{R}^{d}$ \citep{wiener56, granger69}. The effect is traditionally measured through the variance of residuals in restricted and unrestricted models, like shown in Definition \ref{grangerdefinition1}, where $\mathcal{P}$ represents the optimal prediction function, $\mathcal{I}_{<t}$ is all information prior to time $t$, and $\mathbf{X}_{<t}$ is a matrix of compiled values of $\mathbf{x}_t$ prior to time $t$.
	\begin{definition}
		Time series $\mathbf{x}_t\in\mathbb{R}^{p}$ is Granger causal for $\mathbf{y}_t\in\mathbb{R}^{d}$ if 
		\begin{align}
			\text{Var}\left[\mathbf{y}_t - \mathcal{P}\left(\mathbf{y}_t|\mathcal{I}_{<t}\right)\right] < \text{Var}\left[\mathbf{y}_t - \mathcal{P}\left(\mathbf{y}_t|\mathcal{I}_{<t}\backslash \mathbf{X}_{<t}\right)\right].
		\end{align}
		\label{grangerdefinition1}
	\end{definition}
	\par
	Modern methods for adapting Granger causality to nonlinear functional relationships leverage deep learning and representation learning for capturing dependence between variables. These tools, when brazenly paired with other machine learning techniques, are not necessarily reliable or precise. Penalized variable selection as a screening method forces the isolation of information; this does not consider the collective covariate set of the system and tends to up-weight contributions of the chosen nonzero covariates. Erroneous conclusions can result from deep learning when variable specific inference is muddled. This work explicitly redefines Granger causality in terms of a permuted framework with out-of-sample testing (NPGC) that retains the flexibility of representation learning and has specific advantages when seeking nonlinear functional connections. 
	
	\subsection{Granger Causality}\label{se:GC}
	The form of Granger causality presented in Definition \ref{grangerdefinition1} is inherently conditional on additional information included in the modeling process. The optimal prediction $\mathcal{P}\left(\mathbf{y}_t|\mathcal{I}_{<t}\right)$ is unattainable in practice, and the notion of Granger causality is a conditional model on some included explanatory covariate set $\mathbf{z}_t\in\mathbb{R}^{q}$ and the history of the response prior to time $t$, $\mathbf{Y}_{<t}$, like that in Definition \ref{grangerdefinition2}. 
	\begin{definition}
		Time series $\mathbf{x}_t\in\mathbb{R}^{p}$ is conditionally Granger causal for $\mathbf{y}_t\in\mathbb{R}^{d}$ given $\mathbf{z}_t\in\mathbb{R}^{q}$ and the relevant history of the response if 
		\begin{align}
			\text{Var}\left[\mathbf{y}_t - \mathcal{P}\left(\mathbf{y}_t|\mathbf{Y}_{<t},\mathbf{Z}_{<t},\mathbf{X}_{<t}\right)\right] < \text{Var}\left[\mathbf{y}_t - \mathcal{P}\left(\mathbf{y}_t|\mathbf{Y}_{<t},\mathbf{Z}_{<t}\right)\right].
		\end{align}
		\label{grangerdefinition2}
	\end{definition}
	\par
	Model selection is implicit in the test for presence of Granger causality, and the inferential conclusion is coupled with a written form \citep{friston13}. Inclusion of additional variables strengthens the condition for establishing Granger causality; rejection of the null implies the covariate set $\mathbf{x}_t$ is found to provide unique and useful information for prediction of $\mathbf{y}_t$ beyond that contained in both $\mathbf{z}_t$ and the lagged response \citep{granger80}. The basis for Granger causality requires fulfillment of several conditions including a sufficient length of continuous-valued stationary data, exact and complete specification of the model, error-free observation of the variables, and a sampling frequency on a regular discrete grid that contains the known lag relationship \citep{granger69, granger80, granger88}. In the form of Definition \ref{grangerdefinition2}, the second condition can be relaxed provided that inference is accordingly narrowed to the conditional statement.
	\par
	The framework can also be defined in terms of non-causality as a statement of conditional independence, where inclusion of additional variables $\mathbf{z}_t$ is a simple extension \citep{granger80, florens82}. 
	\begin{definition}
		Time series $\mathbf{x}_t\in\mathbb{R}^{p}$ does not Granger cause $\mathbf{y}_t\in\mathbb{R}^{d}$ if and only if 
		\begin{align}
			\mathbf{Y}_{<t+1}\perp \mathbf{X}_{<t}\hskip0.1in\text{given}\hskip0.1in \mathbf{Y}_{<t}.
		\end{align}
		\label{grangerdefinition3}
	\end{definition}
	\par
	This statement is equivalent, but perhaps more powerful because it is defined in terms of the distributions of the variables, allowing for extension to several other forms of statistical tests. It may be prone to misuse if interpreted to place the burden of proof on establishing independence. In a definition from Section 4 of \citet{shojaie22}, column $j$ in time series $\mathbf{x}_t$ is Granger non-causal for time series $\mathbf{y}_t$ if and only if $\forall t$, 
	\begin{align}
		\mathcal{P}\left(\mathbf{y}_{t}|\mathbf{x}_{<t1},\ldots,\mathbf{x}_{<tj},\ldots,\mathbf{x}_{<tp}\right)= \mathcal{P}\left(\mathbf{y}_{t}|\mathbf{x}_{<t1},\ldots,\mathbf{x}_{<t(j-1)},\mathbf{x}_{<t(j+1)},\ldots,\mathbf{x}_{<tp}\right)\label{noncausal}
	\end{align}
	that implies the equality of these predictions holds for $\emph{all}$ time points $t$. If misinterpreted to mean finding correlation at \emph{one} time point in series $\mathbf{y}_t$ is enough to claim functional dependence of two time series, the presence of a causal connection effectively becomes the null hypothesis and Granger causality is reduced to an exceptionally weak statement.
	\par
	Even when the question is framed with the onus on demonstrating a functional relationship, all Granger causal methods overreach their scope of reliable application when inference is performed on the individual variables included in the covariate set rather than on their collective behavior. Applications of the framework to interpreting individual model coefficients introduce a hidden multiplicity problem of repeated testing on subsets of $\mathbf{x}_t$, and the conclusion requires amendment to conditional non-causality of $\mathbf{y}_t$ given an exhaustive list of all other components $\mathbf{x}_{tj}$, $j=1,\ldots,p$ in the model after adequate Type 1 error control. Methods that select causal covariate pairs through the use of penalized optimization problems don't always allow for easy extension to the multiple testing problem, and may require repetitive sub-sampling approaches such as stability selection \citep{meinshausen10}. A thorough definition of Granger causality provides a clear representation of the collective conclusion to be drawn on the covariate set $\mathbf{x}_t$, clarifies the conditional nature of the result on the specified model and included variables, and stresses the philosophical ordering from the null hypothesis implying no causal structure to the alternative that demands evidence of the contrary, all while retaining any general functional form of $\mathcal{P}\left(\mathbf{y}_t|\mathbf{Y}_{<t}, \mathbf{Z}_{<t},\mathbf{X}_{<t}\right)$.
	\par
	\citet{holland86} relates the definition to that of \citet{suppes70} and criticizes its fragile reliance on the specified pre-exposure variables that may completely change an inferential result. \citet{maziarz15} writes that ``Granger causality does not meet the requirements of an investigator who uses this method due to epistemic reasons'' and the methodology should be used ``only if the theoretical background is insufficient," noting the common cause fallacy, indirect causality, and problems related to sampling frequency. In this tone, the predictive nature of these definitions can relate to a causal structure between two variable groups, but alone is not enough to establish \emph{effective} connectivity, distinguishing a direct influence of one population on another \citep{bressler11, friston94}. Even in the presence of the optimal set $\mathcal{I}_{<t}$, association and precedence are not enough to distinguish true causality if slight redundancies are included or an effect does not remain constant in direction through time \citep{maziarz15}. Granger causality exists in the realm of \emph{functional} connectivity that identifies correlation at one or more time lags \citep{friston94}. Appropriate use of Granger causality is contentious, but the method has been applied to a variety of fields like economics, environmental sciences, and neuroscience \citep{bernanke90, cox15, holland86, reid19, seth15, sims72}. Cautious and targeted use of the Granger causal framework can elucidate predictive relationships between variables that warrant further study when prior knowledge of potential causal relationships is limited.
	\par
	In the linear realm, $\mathcal{P}(\mathbf{y}_t|\mathbf{Y}_{<t},\mathbf{Z}_{<t},\mathbf{X}_{<t})$ is often sought from a VAR model and evaluated with in-sample testing \citep{granger69, granger80}. \citet{geweke82} proposed a spectral decomposition form of linear Granger causality for application to stationary Gaussian processes. Inference is performed with the estimated covariance of the restricted model (with $\mathbf{X}_{<t}$ excluded) and that of the unrestricted model, where under the null it is assumed the two are equal. Increasing dimension of the response variable often requires implementation of an approximate test or a switch to permutation-like testing \citep{anderson01, barnett11}. In-sample testing differs from the true notion of predictive ability, and out-of-sample methods align closer to the essence of Granger causality \citep{chao01, inoue05, peters16}. This distinction is especially important when using deep learning techniques to model complex, nonlinear dynamics. 
	
	\subsection{Nonlinear Adaptations}\label{se:NLGC}
	Nonparametric methods provide the basis for many nonlinear adaptations of Granger causality; specification of the exact functional form can be challenging. Parametric attempts, like the ordinary differential equations approaches of \citet{henderson14} and \citet{wu14}, allow for flexible definition of a series of functions to capture dependence, but are limited to modeling additive dynamics when the true mechanism of interaction may be more complicated. Some model-free information theoretic methods detect more elaborate forms of nonlinear dependence with minimal assumptions, but suffer from highly variable estimates and challenges of application to multivariate systems \citep{amblard11, runge12, vicente11}. Kernel Granger causality examines the linear form in a transformed feature space, but model comparison can become difficult \citep{marinazzo08, marinazzo11}. Artificial neural networks (ANNs) allow for general forms of nonlinear dependence in a similar feature space.
	\par
	Fully connected feed-forward networks (with $g$ a sigmoid function), like the construction shown in Equations \ref{eq:deephidden1} through \ref{eq:deepoutput}, fall under the umbrella of the universal approximation theorems of \citet{cybenko89} and \citet{hornik91}. These artificial networks act as featurizers of the data to a high-dimensional space, $\Psi: \mathbf{X}\in\mathbb{R}^{T\times p} \rightarrow \mathbf{H}\in\mathbb{R}^{T\times N}$, $N\gg p$; the architecture allows for arbitrarily close approximation $\hat{\mathbf{y}}_t$ of any function output $\mathbf{y}_t=f(\mathbf{x}_t)$, provided $\mathbf{h}_t\in\mathbb{R}^N$ is allowed to contain a sufficient width of hidden states $N$. 
	\begin{align}
		\mathbf{h}^1_t &= g\left(\mathbf{W}^0\mathbf{x}_t + \mathbf{b}^0\right)\label{eq:deephidden1}\\
		\mathbf{h}^{\ell+1}_t &= g\left(\mathbf{W}^\ell\mathbf{h}^\ell_t + \mathbf{b}^\ell\right)\label{eq:deephiddenl}\\
		\hat{\mathbf{y}}_t &= \mathbf{W}^L\mathbf{h}^L_t + \mathbf{b}^L\label{eq:deepoutput}
	\end{align}
	In Equations \ref{eq:deephidden1} and \ref{eq:deephiddenl}, $g$ is an element-wise activation function, and $\ell=1,\ldots,L$ is the number of layers, or the ``depth'' of the network, where $\mathbf{W}^0\in\mathbb{R}^{N\times p}$, $\mathbf{W}^\ell\in\mathbb{R}^{N\times N}$, $\mathbf{W}^L\in\mathbb{R}^{d\times N}$,  $\mathbf{b}^0,\mathbf{b}^\ell\in\mathbb{R}^{N}$, and $\mathbf{b}^L\in\mathbb{R}^d$ are parameter matrices. The artificial networks shown above, sometimes called multilayer perceptrons (MLPs), and recurrent neural networks (RNNs) like long-short term memory networks (LSTMs) \citep{hochreiter97, graves12} and echo state networks (ESNs) \citep{jaeger07} have all been used to investigate functional connectivity in time series data \citep{shojaie22}. RNN-type network construction adds a recurrence term to the hidden state form in Equations \ref{eq:deephidden1} and \ref{eq:deephiddenl}, $\mathbf{h}^1_t = g\left(\mathbf{W}^{1h}\mathbf{h}^1_{t-1} +\mathbf{W}^0\mathbf{x}_t + \mathbf{b}^0\right)$, and $\mathbf{h}^{\ell+1}_t = g\left(\mathbf{W}^{\ell h}\mathbf{h}^{\ell+1}_{t-1} +\mathbf{W}^\ell\mathbf{h}^{\ell}_{t} + \mathbf{b}^\ell\right)$. 
	\par
	\citet{tank22} extend Granger causality to the nonlinear space using component-wise MLPs (cMLP), which model individual variables in the response $\mathbf{y}_{ti}$, $i=1,\ldots,d$, with separate artificial networks. Parameters are sought via a penalized optimization approach and proximal gradient descent, and no Granger causal connection is inferred for an individual covariate if the corresponding row in a component input parameter matrix ($\mathbf{W}^0_i$ in Equation \ref{eq:deephidden1}) is zero \citep{tank22}. The penalized approach encourages sparse solutions that block the inclusion of information from less predictive components in the hidden states $\mathbf{h}_t$, but selecting the regularization parameter is not an easy task and values may produce vastly different results. \citet{tank22} further introduce a component-wise LSTM (cLSTM) model that harnesses the recurrent structure to circumvent selection of the optimal lag for inclusion in the covariate set $\mathbf{X}_{<t}$. This formulation, while making model specification as it relates to the time lag components easier, has the consequence of mixing inferential results across several lags. 
	\par
	Other methods for capturing arbitrary, nonlinear functional relationships include \citet{khanna19} that builds on the framework of cMLP with statistical recurrent units, \citet{biswas22} that discusses the application of the component network structure to frequency-specific relationships and non-stationary data, and \citet{marcinkevics21} that introduces generalized vector autoregressive (GVAR) methodology aimed at interpretability of potential functional relationships. The Jacobian Granger causality method of \citet{suryadi23} uses the Jacobian matrix, and \citet{nauta19} (TCDF) uses convolutional neural networks (CNNs) and attention scores to serve as measures of variable importance. Jointly estimating a large number of parameters is computationally expensive, and \citet{duggento21} instead use randomly initialized ESNs. Because computation is performed using linear techniques rather than a gradient descent algorithm, complexity decreases; however, inference can only be performed on the output coefficients if information mixing does not occur in the hidden states $\mathbf{h}_t$. They formulate each $\mathbf{W}^h$ in the RNN structure as a block diagonal matrix, which limits the scope of application to a specific subset of additive nonlinear interactions \citep{duggento21}.
	\par
	Many of these methods ignore the collective inference principle of Granger causality and instead take the eager approach of performing individual covariate inference, sometimes with disregard for the multiplicity problem. Evaluation of these predictive relationships is often performed using in-sample tests. With their universal approximator ability, artificial neural networks of sufficient width or depth can approximate \emph{any} functional relationship between two covariate sets, even if it is data-specific and the model is overfit, making them prone to link variables that do not have a predictive relationship as the dimension of the network increases. Sparsity inducing penalties may marginally improve reliability of in-sample tests, but out-of-sample testing helps control the overfitting problem to identify only \emph{useful} functional relationships. In this vein, \citet{horvath22} develop the Learned Kernel VAR (LeKVAR) method that proposes use of a kernel parameterized by an artificial neural network they argue is less prone to overfitting from a decoupling importance measure of the individual series and the selected lags, and the TCDF method employs a permutation-like testing procedure after variable selection \citep{nauta19}. 
	\par
	Rather than comparison of the inherently unequal restricted and unrestricted model errors, focus is shifted to out-of-sample predictability by implementing a permutation structure. There is precedence for the use of permutation-type procedures on general linear models, and their asymptotics are well studied in literature \citep{anderson99, anderson01, diciccio17, winkler14}. \citet{nauta19} employ this type of importance measure after their complicated CNN screening procedure, this work explicitly defines the methodology for its widespread use as a decision framework in Granger causal inference. Importance of a chronologically ordered variable can be interpreted as ``causal'' (predictive) effect, and the strategy builds on the concept of exchangeability (like the directed graph method of \citet{caron17}), where if $\tilde{\mathbf{X}}$ is a random permutation of the rows of $\mathbf{X}$, $\mathbf{Y}_{<t+1} \perp \mathbf{X}$ given $\mathbf{Y}_{<t}$ implies $\mathbf{Y}_{<t+1} \perp \tilde{\mathbf{X}}$ given $\mathbf{Y}_{<t}$, but the converse is not always true \citep{van06}.
	\par
	The following definition pair, adjusting Definitions \ref{grangerdefinition2} and \ref{grangerdefinition3} to a permutation structure for the covariate matrix, are proposed to investigate if $\mathbf{x}_t$ Granger causes $\mathbf{y}_t$. The unrestricted and restricted models are replaced by a null model and a permuted model, where $\tilde{\mathbf{X}}_{<t}$ is a copy of $\mathbf{X}_{<t}$ with the time axis (rows) randomly permuted.
\definitiongroup
\begin{subdefinition}
		Time series $\mathbf{x}_t\in\mathbb{R}^{p}$ is not conditionally Granger causal for $\mathbf{y}_t\in\mathbb{R}^{d}$ given $\mathbf{z}_t\in\mathbb{R}^{q}$ and the relevant history of the response $\mathbf{Y}_{<t}$ if and only if 
		\begin{align}
			\mathbf{Y}_{<t+1}|\mathbf{Y}_{<t}, \mathbf{Z}_{<t}, \mathbf{X}_{<t} \overset{d}{=} \mathbf{Y}_{<t+1}|\mathbf{Y}_{<t}, \mathbf{Z}_{<t}, \tilde{\mathbf{X}}_{<t}.
		\end{align}
		\label{grangerdefinition4a}
	\end{subdefinition}
	\begin{subdefinition}
		Equivalent to Definition \ref{grangerdefinition4a}, $\mathbf{x}_t$ is conditionally Granger causal for $\mathbf{y}_t$ given $\mathbf{z}_t$ and the relevant history of the response $\mathbf{Y}_{<t}$ if 
		\begin{align}
			\text{Var}\left[\mathbf{y}_t - \mathcal{P}\left(\mathbf{y}_t|\mathbf{Y}_{<t},\mathbf{Z}_{<t},\mathbf{X}_{<t}\right)\right] < \text{Var}\left[\mathbf{y}_t - \mathcal{P}\left(\mathbf{y}_t|\mathbf{Y}_{<t},\mathbf{Z}_{<t},\tilde{\mathbf{X}}_{<t}\right)\right].
		\end{align}
		\label{grangerdefinition4b}
	\end{subdefinition}
	\par
	The permutations of $\mathbf{X}_{<t}$ break the potential dependence structure with the response while retaining the intradependence of the covariates, as information in a given row across columns remains intact. The restructuring of the Granger causal framework allows for use of out-of-sample estimated prediction errors, corrects the imbalance of comparison between restricted and unrestricted, and presents a clear path to account for multiple testing, aligning the methodology closer to its inferential utility.
	
	\section{Methodology}
	\label{se:NPGCmethods}
	Suppose observed realizations of the data $\left(\mathbf{X}, \mathbf{Y}, \mathbf{Z}\right)_{\omega}$ arise from the set of all potential realizations $\omega\in\Omega$. Define $\Omega_\text{obs}$ as the size $\varphi$ set of observations, $\Omega_\text{obs}=\left\{1,\ldots,\varphi\right\} \subseteq \Omega$, and note that usually $\varphi=1$. Instances for $\varphi>1$ may occur with repeated trials of a controlled experiment. For simplicity of the original presentation, the subscript $\omega$ notation specifying an observed realization is omitted until the end of this subsection.
	
	\subsection{Structure}\label{ss:Structure}
	Nonlinear functional dependence in the data is captured with feed forward networks (FNNs) like Equations \ref{eq:deephidden1} and \ref{eq:deepoutput}, where the depth $L=1$, and $\mathbf{W}^0$ and $\mathbf{b}^0$ are randomly generated. The exact formulation of this transformation to a representative space (akin to $\Psi$ in Section \ref{se:NLGC}) is not the main focus of this work, and the dimension of the feature space $N$ is fixed for direct comparison to other methods. For demonstration of NPGC, a simple model agnostic structure was selected. Other, more targeted featurization strategies will likely more effectively describe data specific dependence.
	\par
	With familiarity of a dataset, a researcher selects $\gamma$ lagged values of $\mathbf{y}_t$ that serve as a representative history, $\mathbf{Y}_\text{lag}\in\mathbb{R}^{(T + \gamma) \times \gamma d}$ (corresponding to $\mathbf{Y}_{<t}$ in Definitions \ref{grangerdefinition4a} and \ref{grangerdefinition4b}), and assume that an appropriate number of lags is selected such that autocorrelation in $\mathbf{y}_t$ is fully explained across all potential realizations. Selection of the truncation lag $\gamma$ is outside the scope of this work; \citet{ng01, ivanov05, shojaie10, nicholson17} provide detailed discussions. With a finite data length $T + \gamma$, this process restricts the usable portion of $\mathbf{X}$, $\mathbf{Y}$, and any additional covariates $\mathbf{Z}$ to the last $T$ rows. All columns (individual variables) are standardized for consistent behavior in a random FNN featurization process with an activation function \citep{goodfellow16}.
	\par
	The rows of the covariate matrix $\mathbf{X}$ are randomly reorganized via $\bm{\Pi}_m$ to generate several permutations $\tilde{\mathbf{X}}_{m}=\bm{\Pi}_m\mathbf{X}$ for $m = 1, \ldots, M$. The designated first permutation, $m=1$, corresponds to the original ordering of the data where $\bm{\Pi}_1 = \mathbf{I}$. Dependence across rows of $\mathbf{X}$, for example a covariate lag structure, is captured by augmenting the matrix with additional columns.
	\par
	After permutation, the predictor matrices $\left[\mathbf{1}\;\mathbf{Y}_\text{lag}\;\mathbf{Z}\;\tilde{\mathbf{X}}_m\right]\in\mathbb{R}^{T\times (1 + \gamma d + q + p)}$ are compiled, and the FNN is rewritten to the structure in Equation \ref{featurizer} (with activation function $g=\tanh$). $\mathbf{W}^0$ and $\mathbf{b}^0$ are combined into a single parameter matrix $\mathbf{W}\in\mathbb{R}^{(1 + \gamma d + q + p)\times N}$ after inclusion of the intercept term in the predictor matrices, and each matrix entry is an independent Gaussian realization $w_{ij}\sim\mathcal{N}(0,1)$.
	\begin{align}
		\mathbf{H}_m &=g\left(\left[\mathbf{1}\;\mathbf{Y}_\text{lag}\;\mathbf{Z}\;\tilde{\mathbf{X}}_m\right]\mathbf{W}\right) = \tanh\left(\left[\mathbf{1}\;\mathbf{Y}_\text{lag}\;\mathbf{Z}\;\tilde{\mathbf{X}}_m\right]\mathbf{W}\right)\label{featurizer}
	\end{align}
	\par
	Added uncertainty arising from the random generation is mitigated through several featurizations, $\mathbf{W}_r$ for $r=1,\ldots,\mathscr{R}$, and extracting the aggregate behavior. The models can be written in terms of the original ($m=1$) and permuted ($m=2,\ldots,M$) feature spaces, where $\mathbf{U}_{m,r}$ is the variation in $\mathbf{Y}$ not captured by the functional relationship with the feature space $\mathbf{H}_{m,r}$. 
	\begin{align}
		\mathbf{Y} = \mathbf{H}_{m,r}\mathbf{W}^L_{m,r} + \mathbf{U}_{m,r}
	\end{align}
	\par
	Define $\bm{\Theta}_m$ as the underlying covariance matrix of the prediction for $\mathbf{Y}$ given the relevant history of the response $\mathbf{Y}_\text{lag}$, the additional variables $\mathbf{Z}$, and the permuted covariate set $\tilde{\mathbf{X}}_m$. Denote $\vartheta_m=\text{tr}\left(\bm{\Theta}_m\right)$ as the corresponding parameter over all potential realizations $\omega\in\Omega$. Variation in the estimate arises from potential realizations of the data $\omega\in\Omega$ and via randomly generated FNNs approximating the nonlinear functional form. For a given data realization $\left(\mathbf{X}, \mathbf{Y}, \mathbf{Z}\right)_\omega$, and under the true functional form $f$, define the specific covariance matrix of the prediction $\bm{\Sigma}_{m,\omega}$. For random featurization $r$, define the covariance matrix of prediction $\mathbf{S}_{m,\omega,r}$ as an estimate of $\bm{\Sigma}_{m,\omega}$.  The out-of-sample variation parameter $\vartheta_m$ is estimated for each permutation via a cross-validation approach.
	
	\subsection{Estimating Granger Causal Influence}\label{ss:ECGI}
	A sufficiently large value featurization dimension $N$ is chosen to ``linearize'' any existing functional relationship, but not so large that the network is able to memorize inputs or fabricate dependence between the permuted data and a response. This implicitly assumes the existence of some nonlinear functional relationship between a covariate set and a response will be ``easier'' for an ANN to learn than random matching of inputs to outputs in the permuted data, and an exact form of this condition is proposed in Section \ref{se:NPGCtheory}. The data is split into $K$ sets for model computation and testing, and define the number of observations in each set $k=1,\ldots,K$ as $T_{k} = \lfloor T/K\rfloor + \mathbf{1}\left\{\left(T\bmod K\right)\geq k\right\}$.
	\par
	Under the form of Equation \ref{featurizer}, several random FNNs are generated $r=1,\ldots,\mathscr{R}$. The model matrices $\mathbf{W}_r$ are held fixed over all permutations $m$ and observations $\omega\in\Omega_\text{obs}$ for a consistent error estimation framework. For each permutation, with $m=1$ corresponding to the original data, the predictor matrices are projected into the respective feature spaces, and the residuals for test set $k$ can be written as in Equation \ref{testresiduals}, where the where the training data (subscript $-k$) excludes set $k$. 
	\begin{align}
		\mathbf{R}_{m,\omega,r,k}&= \mathbf{H}_{m,\omega,r,k} \left(\mathbf{H}_{m,\omega,r,-k}'\mathbf{H}_{m,\omega,r,-k}\right)^{-1}\mathbf{H}_{m,\omega,r,-k}'\mathbf{Y}_{\omega,-k}, - \mathbf{Y}_{\omega,k}\label{testresiduals}
	\end{align}
	\par
	The out-of-sample prediction residuals $\mathbf{R}$ from the test set align closer to the original definition of predictive ability in Granger causal inference than the in-sample model variation. This distinction is especially important with the use of ANNs and the ability to learn any arbitrary, data-specific dependence structure. The estimate for the out-of-sample variation in prediction residuals $\hat{\vartheta}_{m}$ is shown in Equation \ref{parameter}.
	\begin{align}
		\hat{\vartheta}_{m} &= \frac{1}{\varphi\mathscr{R}K}\sum_{\omega = 1}^{\varphi} \sum_{r=1}^{\mathscr{R}} \sum_{k=1}^K \frac{1}{T_k} \text{tr}\left(\mathbf{R}_{m,\omega,r,k}'\mathbf{R}_{m,\omega,r,k}\right)\label{parameter}
	\end{align}
	\par
	A null distribution of variation from each model $m=1,\ldots,M$ is approximated from the permutation structure. The random permutations of the covariate set $\mathbf{X}$ break the potential dependence structure present in the form of some predictive relationship with the response $\mathbf{Y}$. Under the null hypothesis, the original data is viewed as one of $M$ random permutations, and the original ``permutation'' $\tilde{\mathbf{X}}_1=\mathbf{X}$ should exhibit similar properties to $\tilde{\mathbf{X}}_m$ for $m=2,\ldots,M$. Analogously, if the time observations $\mathbf{X}$ are exchangeable for prediction of $\mathbf{Y}$, the conditional distribution of the prediction will not change.
	\par
	The variation estimates are drawn from the distribution of all possible permutations in Equation \ref{fullpermutationdistribution}, and $\hat{\vartheta}_{1}$ is expected to fall above some lower tail portion. 
	\begin{align}\label{fullpermutationdistribution}
		\hat{\vartheta}_{m}\sim\hat{\mathcal{H}}(s) &= (T!)^{-1}\sum_{i=1}^{T!}\mathbf{1}\{\hat{\vartheta}_{i}\leq s\},
	\end{align}
	This leads to an approximate null distribution where the estimate $\hat{\vartheta}_{1}$ is at quantile $\hat{Q}_M$ of the empirical distribution $\hat{\mathcal{H}}_{M}(s)$, defined in Equation \ref{samplepermutationdistribution}, formed from a subsample of size $M\leq T!$. A decision rule is formulated from comparison to a chosen level of test $\alpha$. 
	\begin{align}
		\hat{\mathcal{H}}_{M}(s) &= \frac{1}{M}\sum_{m=1}^{M}\mathbf{1}\{\hat{\vartheta}_{m}\leq s\}\label{samplepermutationdistribution}\\
		\hat{Q}_{M} = \hat{\mathcal{H}}_{M}(\hat{\vartheta}_{1})&=\frac{1}{M}\sum_{m=1}^M\mathbf{1}\left\{\hat{\vartheta}_{m} \leq \hat{\vartheta}_{1}\right\}\label{quantileestimate}
	\end{align}
	\par
	Rejection of the null hypothesis in this framework, $\hat{Q}_M \leq \alpha$, presents evidence for $\mathbf{X}$ as Granger causal of $\mathbf{Y}$ \emph{conditional} on the additional variables $\mathbf{Z}$ and the relevant history of the response $\mathbf{Y}_\text{lag}$. For the case when $\varphi=1$, this conclusion is conditional on \emph{error free observation} of the dataset. Inferential results must either include this assumption, or the scope narrowed to the specific observation $\omega$. The full algorithmic process is shown in the supplement, and an explicit outline of the theoretical behavior of these estimates and development of the underlying framework is given in Section \ref{se:NPGCtheory}.

	\section{Theory}\label{se:NPGCtheory}
	Define $\bm{\Theta}_m$ as the underlying covariance matrix of the predictive ability of permutation $m$, and the quantity $\vartheta_m=\text{tr}\left(\bm{\Theta}_m\right)$. For each potential realization of the data $\left(\mathbf{X}, \mathbf{Y}, \mathbf{Z}\right)_\omega$, $\omega\in\Omega$, define the realization-specific covariance matrix $\bm{\Sigma}_{m,\omega}$ drawn from some distribution with expectation $\mathbb{E}\left[\text{tr}\left(\bm{\Sigma}_{m,\omega}\right)\right] = \vartheta_m$ and variance $\tau_\omega^2<\infty$ that is constant over all permutations. Each random generated FNN for the featurization process $r=1,\ldots,\mathscr{R}$ produces $\mathbf{S}_{m,\omega,r}$ as an estimate of $\bm{\Sigma}_{m,\omega}$, where $\mathbb{E}\left[\text{tr}\left(\mathbf{S}_{m,\omega,r}\right)|\bm{\Sigma}_{m,\omega}\right] = \text{tr}\left(\bm{\Sigma}_{m,\omega}\right)$ and $\text{Var}\left[\text{tr}\left(\mathbf{S}_{m,\omega,r}\right)|\bm{\Sigma}_{m,\omega}\right] = \tau_r^2 < \infty$. 
	\par
	The null and permuted models are evaluated with the out-of-sample prediction residuals, shown in Equation \ref{testresiduals}, and define the estimate for total variation as in Equation \ref{parameter}. Under the null hypothesis when the conditional distribution of the prediction is invariant to permutation of the covariate set $\mathbf{X}$, $\vartheta_1=\vartheta_2=\cdots=\vartheta_{T!}=\vartheta$, leading to the null and alternative hypotheses in Equation \ref{finalhypothesis}.
	\begin{align}
		H_0 &: \vartheta_1=\vartheta_2=\cdots=\vartheta_{T!-1}=\vartheta_{T!}\nonumber\\
		H_A &: \vartheta_1<\vartheta_i\hskip0.1in\text{for all $i=2,\ldots,T!$}\label{finalhypothesis}
	\end{align}
	\par
	The null hypothesis is tested using the sample quantile $\hat{Q}_M$ from the empirical distribution $\hat{\mathcal{H}}_{M}(s)$ defined in Equations \ref{samplepermutationdistribution} and \ref{quantileestimate}. Theoretical results in this section rely on a set of three mild conditions comparable to those found in relevant literature. Four additional conditions provide regularity to the featurization process.
	\begin{condition}\label{samplingcondition}
		The data is continuous and stationary, and the discrete, regular sampling grid $t=1,\ldots,T$ is sufficiently fine to capture any potential functional dependence in the variable matrices $\left(\mathbf{X}, \mathbf{Y}, \mathbf{Z}\right)_\omega$ for any realization $\omega\in\Omega$.
	\end{condition}
	\begin{condition}\label{distributioncondition}
		For all realizations $\omega\in\Omega$, and random generated FNNs $\mathbf{W}_r\in\mathcal{W}$, $r=1,\ldots,\mathscr{R}$, the quantities $\text{tr}\left(\mathbf{S}_{m,\omega,r}\right)|\bm{\Sigma}_{m,\omega}$ are independently drawn from continuous distributions with defined expectation in Equation \ref{eq:distribcond1} and constant, finite variance across all permutations and potential realizations of the data.
		\begin{align}
			\mathbb{E}\left[\text{tr}\left(\mathbf{S}_{m,\omega,r}\right)|\bm{\Sigma}_{m,\omega}\right] &= \text{tr}\left(\bm{\Sigma}_{m,\omega}\right)\label{eq:distribcond1}\\
			\text{Var}\left(\text{tr}\left[\mathbf{S}_{m,\omega,r}\right]|\bm{\Sigma}_{m,\omega}\right)&=\tau_r^2<\infty
		\end{align}
		Similarly, $\bm{\Sigma}_{m,\omega}$ are independently drawn from continuous distributions with expectation in Equation \ref{eq:distribcond2} and constant, finite variance across permutations.
		\begin{align}
			\mathbb{E}\left[\text{tr}\left(\bm{\Sigma}_{m,\omega}\right)\right]&=\vartheta_m\label{eq:distribcond2}\\
			\text{Var}\left(\text{tr}\left[\bm{\Sigma}_{m,\omega}\right]\right)&=\tau_\omega^2<\infty
		\end{align}
	\end{condition}
	\begin{condition}\label{modelerrorcondition}
		The relevant history of the response, $\mathbf{Y}_{<t,\omega}=\mathbf{Y}_{\text{lag},\omega}\in\mathbb{R}^{T\times \gamma d}$ is appropriately chosen such that the model errors $\mathbf{U}_{m,\omega,r}$ are independent, or $\mathbf{u}_{m,\omega,r,t}\perp \mathbf{u}_{m,\omega,r,t'}$ for any realization $\omega\in\Omega$, permutation $m=1,\ldots,M$, featurization $r=1,\ldots,\mathscr{R}$, and time point $t=1,\ldots,T$ where $t'\neq t$. Further, the model errors follow multivariate normal distributions with mean zero and constant variation, leading to the result
		\begin{align}\label{modelerrorconditionequation}
			\mathbf{u}_{m,\omega,r,t} &\overset{i.i.d.}{\sim} \mathcal{N}_d\left(\mathbf{0}, \mathbf{S}_{m,\omega,r}\right).
		\end{align}
	\end{condition}
	
	\subsection{Asymptotic Properties}\label{ss:AsympNPGC}
	The asymptotic behavior of the estimates and established testing framework is examined under the listed conditions above and in the supplement. Define the underlying variation parameter for permutation $m$ as $\vartheta_m$, and the estimate $\hat{\vartheta}_m$ as in Equation \ref{parameter}. Assume a fixed test set size $T_k$ and allow the training set $T_{-k} = T-T_k$ and number of folds $K$ to grow as $T\rightarrow\infty$.
	\begin{theorem}\label{theorem1}
		Under the listed conditions, with Condition \ref{distributioncondition} modified such that $\tau_\omega^2=0$ (\textit{i.e.}, error free observation of the data), for all $\varepsilon>0$,
		\begin{align}
			\lim_{\mathscr{R}\rightarrow\infty} \lim_{T\rightarrow\infty}\mathbb{P}\left(|\hat{\vartheta}_m - \vartheta_m|\leq \varepsilon \right) = 1.\label{theorem1equation1}
		\end{align}
		Alternatively, under the listed conditions, for all $\varepsilon>0$,
		\begin{align}
			\lim_{\varphi\rightarrow\infty} \lim_{T\rightarrow\infty}\mathbb{P}\left(|\hat{\vartheta}_m - \vartheta_m|\leq \varepsilon \right) = 1.\label{theorem1equation2}
		\end{align}
	\end{theorem}
	\par
	Without error free observation of the data, the underlying variation $\tau_\omega^2$ remains present in the estimate, but shrinks as the number of observations gets large.
	\begin{theorem}\label{theorem2}
		Under the listed conditions, the estimate for the variation parameter $\hat{\vartheta}_{m}$ admits a Central Limit Theorem with respect to the number of observations $\varphi$ in $\Omega_\text{obs}$.
		\begin{align}
			\lim_{\mathscr{R}\rightarrow\infty}\lim_{T\rightarrow\infty}\sqrt{\varphi}\left(\hat{\vartheta}_{m} - \vartheta_m\right)&\xrightarrow{D} \mathcal{N}\left(0, \tau_\omega^2\right)\label{theorem2equation}
		\end{align}
	\end{theorem}
	\par
	As a direct result of Theorem \ref{theorem2}, $\hat{\vartheta}_{m}$ is a consistent estimate for $\vartheta_m$ when $\tau_\omega^2>0$ as $\varphi$ tends to infinity along with the size of the training and test sets. 
	
	\subsubsection{Under the Null Hypothesis}\label{ss:NullNPGC}
	Under the null hypothesis, the underlying variation parameter is constant for every possible permutation, $\vartheta_1 = \cdots = \vartheta_T!$. Define the quantile estimate $\hat{Q}_M$ as in Equation \ref{quantileestimate}, and establish the limiting uniform distribution from the result of Theorem \ref{theorem1}.
	\begin{theorem}\label{theorem3}
		Under the null hypothesis and the listed conditions, 
		\begin{align}
			\lim_{T\rightarrow\infty} \lim_{M\rightarrow T!} \hat{Q}_M \xrightarrow{D} \text{Uniform}(0,1).\label{theorem3equation}
		\end{align}
	\end{theorem}
	
	\subsubsection{Under the Alternative Hypothesis}\label{ss:AltNPGC}
	Define $Q$ as the true quantile for parameter $\vartheta_1$ over all $T!$ possible permutations. In the limit, the quantile estimate defined in Equation \ref{quantileestimate} converges in probability to $Q$.
	\begin{theorem}\label{theorem4}
		Under the alternative hypothesis and the listed conditions, with Condition \ref{distributioncondition} modified such that $\tau_\omega^2=0$ (\textit{i.e.}, error free observation of the data), for all $\varepsilon>0$, 
		\begin{align}
			\lim_{M\rightarrow T!} \lim_{\mathscr{R}\rightarrow\infty} \lim_{T\rightarrow\infty} \mathbb{P}\left(|
			\hat{Q}_M-Q|\leq \varepsilon\right)=1.\label{theorem4equation1}
		\end{align}
		Similarly, under the alternative hypothesis and the listed conditions, for all $\varepsilon>0$, 
		\begin{align}
			\lim_{M\rightarrow T!} \lim_{\varphi\rightarrow\infty} \lim_{T\rightarrow\infty} \mathbb{P}\left(|
			\hat{Q}_M-Q|\leq \varepsilon\right)=1.\label{theorem4equation2}
		\end{align}
	\end{theorem}
	
	\section{Simulation Study}\label{se:NPGCsims}
	The ability of the permutation-based methodology of NPGC (similar to the decision rule formulation of TCDF) to detect the presence of functional connectivity and control for false positive results is evaluated with a consistent FNN framework (like the cMLP nonlinear transformation). All methods featurize the data to dimension $N=100$. Three of the five comparison methods use in-sample methodology and variants of the lasso penalty within cMLP models \citep{tank22, tibshirani96}. The other two employ random FNN generation like NPGC, and in-sample testing from comparison of restricted and unrestricted models or the substitution of random Gaussian noise.
	\par
	A group decision process is considered where the covariate set $\mathbf{X}$ is (or is not) collectively Granger causal for the response $\mathbf{Y}$. Ability to detect a Granger causal result will almost certainly improve with a more sophisticated nonlinear structure; relative differences between methods are demonstrated at a baseline level. A full list of included methods is given in Table \ref{table1}. 
	\begin{table}[htb]
		\centering
		\begin{tabular}{cl}
			\hline
			(\romannumeral1) &  Nonlinear Permuted Granger Causality (NPGC)\\
			(\romannumeral2) &  cMLP, Group Lasso Penalty (cMLP-GL)\\
			(\romannumeral3) & cMLP, Group Sparse Group Lasso Penalty (cMLP-GSGL)\\
			(\romannumeral4) &  cMLP, Hierarchical Lasso Penalty (cMLP-H)\\
			(\romannumeral5) & Restricted vs. Unrestricted Models (R/U)\\
			(\romannumeral6) & Gaussian Noise Substitution (GNS)\\
			\hline
		\end{tabular}
		\caption{Testing frameworks for Granger causal inference simulations.}
		\label{table1}
	\end{table}
	\par
	The lasso type objectives of cMLP methods (\romannumeral2) - (\romannumeral4) penalize nonzero rows of the matrix $\mathbf{W}^0$ in the formulation of Equation \ref{eq:deephidden1} (or $\mathbf{W}$ in Equation \ref{featurizer}). The corresponding null and alternative hypotheses are adjusted to those shown in Equation \ref{lassohypothesis}, and a Granger
	causal connection is present for a covariate set if any corresponding rows of $\mathbf{W}^0$ contain nonzero components \citep{tank22}.
	\begin{align}
		H_0 &: \mathbf{W}^0_j = \mathbf{0} \text{ for all $j$ corresponding to the covariate set $\mathbf{X}$}\nonumber\\
		H_A &: \mathbf{W}^0_j \neq \mathbf{0} \text{ for at least one $j$ corresponding to the covariate set $\mathbf{X}$}\label{lassohypothesis}
	\end{align}
	A group lasso penalty is applied to each variable $j$ in the model matrix $\mathbf{W}^0$ corresponding to the covariate set $\mathbf{X}$ over all time lags $t,\ldots,t-\gamma$ \citep{yuan06}. As in \citet{tank22}, define $\mathbf{W}^0_{j}=\left[\mathbf{W}^0_{j,t}\;\cdots\;\mathbf{W}^0_{j,t-\gamma}\right]$, with each entry as the weights of the model matrix row for variable $j$ and lagged data up to $t-\gamma$. Penalties for the cMLP methods take the general form $\lambda\sum_{j=1}^p\beta\left(\mathbf{W}^0_{j}\right)$, with specific $\beta$ for each defined below.
	\begin{align}
		\beta_{GL}\left(\mathbf{W}^0_{j}\right) &= \left\|\mathbf{W}^0_{j}\right\|_F\label{penalty1}\\
		\beta_{GSGL}\left(\mathbf{W}^0_{j}\right) &= \left\|\mathbf{W}^0_{j}\right\|_F + \sum_{k=0}^\gamma \left\|\mathbf{W}^0_{j,t-k}\right\|_2\label{penalty2}\\
		\beta_{H}\left(\mathbf{W}^0_{j}\right) &= \sum_{k=0}^\gamma \left\|\left[\mathbf{W}^0_{j,t-k}\;\cdots\;\mathbf{W}^0_{j,t-\gamma}\right]\right\|_F\label{penalty3}
	\end{align}
	The group sparse group lasso penalty combines sparsity of included variables and their lagged values like in \citet{simon13}, and the novel hierarchical penalty of \citet{tank22} retains information about the natural ordering of the variables, encouraging solutions where for some lag $k^*$, $k>k^*$ implies $\mathbf{W}^0_{j,t-k}=\mathbf{0}$. A regularization parameter of $\lambda=0.5$ is chosen based on a cross-validation like trial and error approach. Conclusions drawn from these methods can vary greatly depending on the chosen $\lambda$; smaller values retain all matrix entries and larger values penalize the matrix to zero. Effects of varying $\lambda$ are shown for the application in Section \ref{se:NPGCapplication}.
	\par
	Two additional naive methods are included for comparison to NPGC. The in-sample restricted and unrestricted method examines the ratio of model residuals $\hat{\sigma}^2_\text{Res} / \hat{\sigma}^2_\text{Unres}$ in a randomly generated FNN. The in-sample Gaussian noise substitution is methodologically similar, but instead substitutes Gaussian white noise in place of the covariate set $\mathbf{X}$ to examine $\hat{\sigma}^2_\text{Noise} / \hat{\sigma}^2_\text{Unres}$. The three model errors are shown in Equations \ref{res1} to \ref{res3}, where $\mathbf{E} \sim \mathcal{MN}_{T\times p}\left(\mathbf{0}, \mathbf{I}, \mathbf{I}\right)$, and the corresponding $\hat{\sigma}^2$ shown in Equation \ref{sigma2}. 
	\begin{align}
		\mathbf{U}_{\text{Unres},r} &= \mathbf{Y} - \tanh\left(\left[\mathbf{1}\;\mathbf{Y}_\text{lag}\;\mathbf{Z}\;\mathbf{X}\right]\mathbf{W}_r\right) \mathbf{W}^L_{\text{Unres},r}\label{res1}\\
		\mathbf{U}_{\text{Res},r} &= \mathbf{Y} - \tanh\left(\left[\mathbf{1}\;\mathbf{Y}_\text{lag}\;\mathbf{Z}\;\mathbf{0}\right]\mathbf{W}_r\right) \mathbf{W}^L_{\text{Res},r}\label{res2}\\
		\mathbf{U}_{\text{Noise},r} &= \mathbf{Y} - \tanh\left(\left[\mathbf{1}\;\mathbf{Y}_\text{lag}\;\mathbf{Z}\;\mathbf{E}\right]\mathbf{W}_r\right) \mathbf{W}^L_{\text{Noise},r}\label{res3}\\
		\hat{\sigma}_r^2 &= T^{-1}\mathbf{U}_{r}'\mathbf{U}_{r}\label{sigma2}
	\end{align}
	\par
	The individual terms in Equation \ref{sigma2} follow chi-square distributions with degrees of freedom $T-N$, and each ratio for an individual generated FNN is distributed $\mathcal{F}_{T-N,T-N}$. The sum of a large number of these independent (via Condition \ref{modelerrorcondition}) statistics, all random FNNs $r=1,\ldots,\mathscr{R}$, is approximately normal. As $\mathscr{R}\rightarrow\infty$ under the null hypothesis,
	\begin{align}
		\mathscr{R}^{-1}\sum_{r=1}^\mathscr{R} \frac{\hat{\sigma}_{0,r}^2}{\hat{\sigma}_{1,r}^2} \xrightarrow{D} \mathcal{N}\left[\frac{T-N}{T-N-2}, \frac{4(T-N)^2(T-N-1)}{\mathscr{R}(T-N)(T-N-2)^2(T-N-4)}\right].
	\end{align}
	This is used to formulate a naive decision rule with a large number ($\mathscr{R}=1000$) of randomly generated FNNs.
	\par
	The NPGC methodology is performed with $M=400$ permutations, $K=5$ cross validation folds, and $\mathscr{R}=50$ randomly generated FNNs. The methods are successful when correctly flagging a Granger causal result when direct functional dependence is present, or correctly labelling a non-causal result when it is absent, leading to the potential outcomes in Table \ref{table2}. For the NPGC and naive ratio methods, a result is flagged if the quantile estimate of the associated statistic under the null hypotheses is under a specified level $\alpha$. The lasso-type methods lack the direct translation to a traditional hypothesis testing framework.
	\begin{table}[htb]
			\centering
			\begin{tabular}{ccc}
				\hline
				& \multicolumn{2}{c}{Truth} \\
				Result &  Granger causal ($\text{GC} = 1$) & Not causal ($\text{GC} = 0$)\\\hline
				Granger causal ($\text{GC} = 1$) & $\rho_{1} = \sum_{\mathcal{S}}\mathbf{1}\left\{\mathcal{D} = 1\right\} / N_1$ & $\rho_{10}=\sum_{\mathcal{S}}\mathbf{1}\left\{\mathcal{D} = 1\right\} / N_0$ \\
				Not causal ($\text{GC} = 0$) & $\rho_{01}=\sum_{\mathcal{S}}\mathbf{1}\left\{\mathcal{D} = 0\right\} / N_1$ & $\rho_{0} = \sum_{\mathcal{S}}\mathbf{1}\left\{\mathcal{D} = 0\right\} / N_0$ \\\hline
			\end{tabular}
			\caption{Proportion of potential outcomes for the decision $\mathcal{D}$ of each of the methods shown in Table \ref{table1}. The set $\mathscr{S}$ represents the space of all simulations (both $\text{GC}=1$ and $\text{GC}=0$) within a given setting, $N_{1}=\sum_{\mathscr{S}}\mathbf{1}\left\{\text{GC} = 1\right\}$, and $N_{0}=\sum_{\mathscr{S}}\mathbf{1}\left\{\text{GC} = 0\right\}$.}
		\label{table2}
	\end{table}
	\subsection{Simulation Settings}\label{ss:NPGCsimsettings}
	NPGC and the comparator methods in Table \ref{table1} are tested on two nonlinear processes: Lorenz-96 models of \citet{karimi10}, and TAR models introduced by \citet{tong80}. The $p$-dimensional Lorenz-96 model ($p\geq4$) is governed by the continuous differential equation 
	\begin{align}
		\frac{dx_{i,t}}{dt} &= \left(x_{i+1,t}-x_{i-2,t}\right)x_{i-1,t} -x_{i,t} + F,\label{lorenz96}
	\end{align}
	with $i=1,\ldots,p$ and boundary series $x_{-1,t} = x_{p-1,t}$, $x_{0,t} = x_{p,t}$, and $x_{p+1,t} = x_{1,t}$. $F$ is a forcing constant generated as $F\sim\text{Uniform}(5, 20)$ with higher values introducing a larger degree of nonlinear, chaotic behavior. For data generation, a sampling rate of $\Delta t = 0.05$ and a burn-in period of 500 time steps are used. 
	\par
	The TAR(2) model is governed by a similar skeleton to a VAR, but allows for changes in parameters based on the value of a threshold variable \citep{tong80}.
	\begin{align}
		\begin{bmatrix}
			\mathbf{x}_t\\
			\mathbf{x}_{t-1}
		\end{bmatrix} &= 
		\begin{bmatrix}
			\mathbf{A}^{(k)}_1 & \mathbf{A}^{(k)}_2\\
			\mathbf{I} & \mathbf{0}
		\end{bmatrix}
		\begin{bmatrix}
			\mathbf{x}_{t-1}\\
			\mathbf{x}_{t-2}
		\end{bmatrix} + 
		\begin{bmatrix}
			\bm{\varepsilon}^{(k)}_t\\
			\mathbf{0}
		\end{bmatrix}\label{tarmodel}.
	\end{align}
	The threshold is defined on the values in $\mathbf{x}_{t-2}$ with regime $k=1$ corresponding to the case when $\sum_{i=1}^p x_{i,t-2} \leq 0$ and $k=2$ to $\sum_{i=1}^p x_{i,t-2} > 0$. For the TAR(2) process, $\bm{\varepsilon}^{(k)}_t\sim \mathcal{N}_p\left(\mathbf{0},\sigma_k^2\mathbf{I}\right)$ where $\sigma_1=0.5$ and $\sigma_2 = 0.2$. Each $\mathbf{A}^{(k)}_1,\mathbf{A}^{(k)}_2\in\mathbb{R}^{p\times p}$ has elements $\mathbf{A}^{(k)}=(a^{(k)}_{ij})\sim\text{Uniform}(-0.5, 0.5)$ that are thresholded to zero if $|a^{(k)}_{ij}| < 0.1$, and their spectral radii are at most 0.8 to ensure stationary data generation. Like the Lorenz-96 data generation procedure, there is a burn-in period of 500 time steps.
	\par
	Two groups of dependent data observations are generated with $p=6$ (one from each process), containing $\mathbf{x}_{i}$ for $i=1,\ldots,12$. The series $i=1,\ldots,6$ are labelled as the Lorenz-96 process and $i=7,\ldots,12$ as the TAR(2) process. The samples are of length $T=250$, $500$, and $1000$ after burn-in and truncation for lag selection to include in the model, $\gamma = 3$. One index is selected to serve as the response, three to serve as the set of additional variables $\mathbf{Z}$ and three for the covariate set $\mathbf{X}$ depending on a Granger causal designation. A dataset is designated ``Granger causal'' if at least one series in $\mathbf{X}$ has a direct influence on the selected response (\textit{i.e.}, one series in $\mathbf{X}$ is chosen from the same generating process). For each of the settings in Table \ref{tablesim}, 200 datasets are generated for a total of 4800 trials.
	\begin{table}[htb]
			\centering
			\begin{tabular}{cccc}\hline
				Response Variable & $T$ &  \# Causal in $\mathbf{X}$ & \# Causal in $\mathbf{Z}$\\\hline
				Lorenz-96 & $\left\{250,500,1000\right\}$ & $\left\{0, 2\right\}$ & $\left\{0, 2\right\}$\\
				TAR(2) & $\left\{250,500,1000\right\}$ & $\left\{0, 2\right\}$ & $\left\{0, 2\right\}$\\\hline
			\end{tabular}
			\caption{NPGC simulation settings. A variable is labelled \textit{causal} if from the same generating process as the response $\mathbf{Y}$. Variables are randomly selected from within their \textit{causal} or \textit{non-causal} groups.}
		\label{tablesim}
	\end{table}
	
	\subsection{Simulation Results}\label{ss:NPGCsimresults}
	Simulation results are split by Granger causal designation and process of the chosen response variable. The control for false positive results and the ability to detect a Granger causal connection is examined in each setting. AUROC is not considered as it does not implicate a decision rule a priori. Tables \ref{table3a} and \ref{table3b} list the proportion of correctly identified causal relationships $\rho_1$ for TAR(2) and Lorenz-96 response variables, respectively. Results for the finer grid of designations from Table \ref{tablesim} are given in the supplement.
	\begin{table}[htb]
			\centering
			\begin{tabular}{ccccccc}\hline
				& NPGC & cMLP-GL & cMLP-GSGL & cMLP-H & R/U & GNS \\\hline
				$T=250$ & 0.915 & 0.870 & 0.155 & 0.555 & 0.752 & 0.989 \\
				$T=500$ & 0.965 & 0.852 & 0.132 & 0.558 & 0.842 & 0.998 \\
				$T=1000$ & 0.985 & 0.850 & 0.122 & 0.500 & 0.850 & 1.000 \\\hline
			\end{tabular}
			\caption{TAR(2) response compiled simulation results. Proportion of correctly labelled Granger causal outcomes ($\text{GC} = 1$, $\rho_1$ in Table \ref{table2}).}
		\label{table3a}
	\end{table}
	\begin{table}
		\centering
			\begin{tabular}{ccccccc}\hline
				& NPGC & cMLP-GL & cMLP-GSGL & cMLP-H & R/U & GNS \\\hline
				$T=250$ & 0.978 & 0.765 & 0.232 & 0.482 & 0.025 & 1.000 \\
				$T=500$ & 0.988 & 0.728 & 0.160 & 0.442 & 0.001 & 1.000 \\
				$T=1000$ & 0.992 & 0.715 & 0.128 & 0.412 & 0.020 & 1.000 \\\hline
			\end{tabular}
			\caption{Lorenz-96 response compiled simulation results. Proportion of correctly labelled Granger causal outcomes ($\text{GC} = 1$, $\rho_1$ in Table \ref{table2}).}
		\label{table3b}
	\end{table}
	\par
	Figures \ref{figure1a} and \ref{figure1b} plot the observed Type 1 error by specified level of test, with a target 45 degree line included for reference. 
	\begin{figure}[htb]
			\centering
			\includegraphics[width = \textwidth]{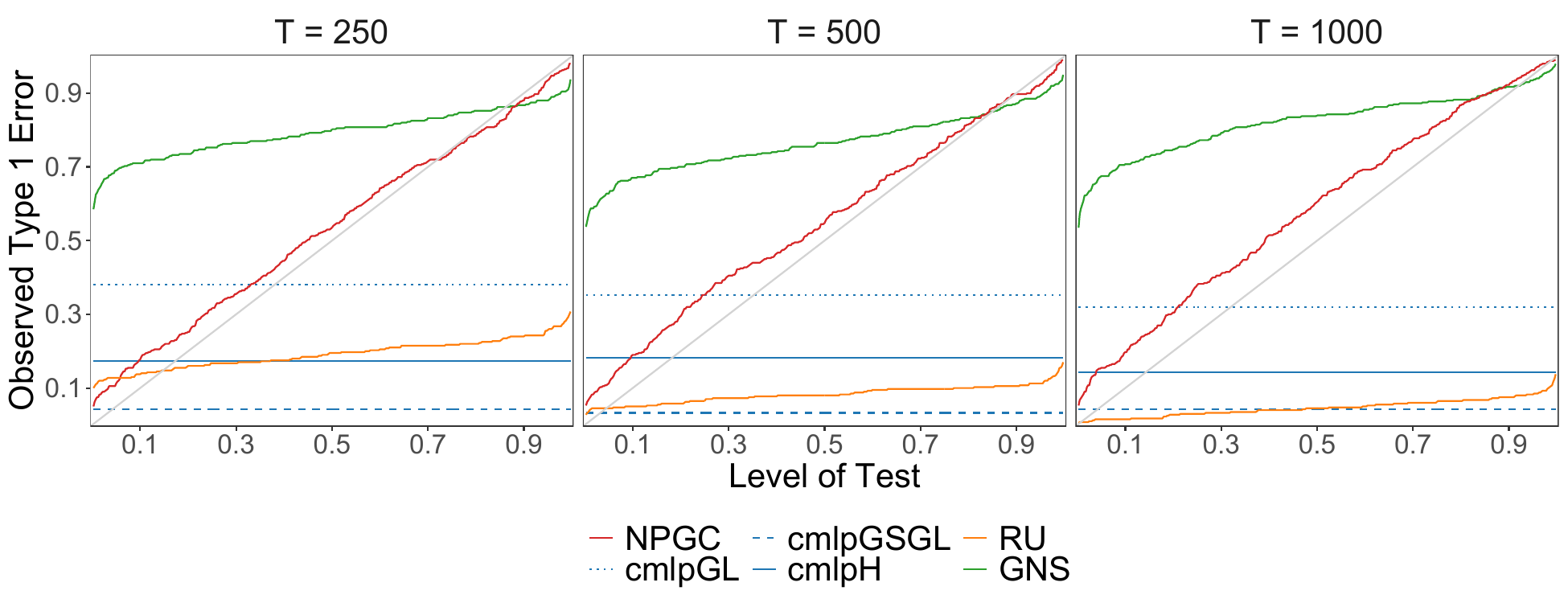}
			\caption{Type 1 error control for TAR(2) simulations. Methods that adhere to a chosen level of test will lie closer to the uniform CDF (gray) included for clarity that indicates the reference quantile $\alpha$.}
			\label{figure1a}
	\end{figure}
	\begin{figure}[htb]
			\centering
			\includegraphics[width = \textwidth]{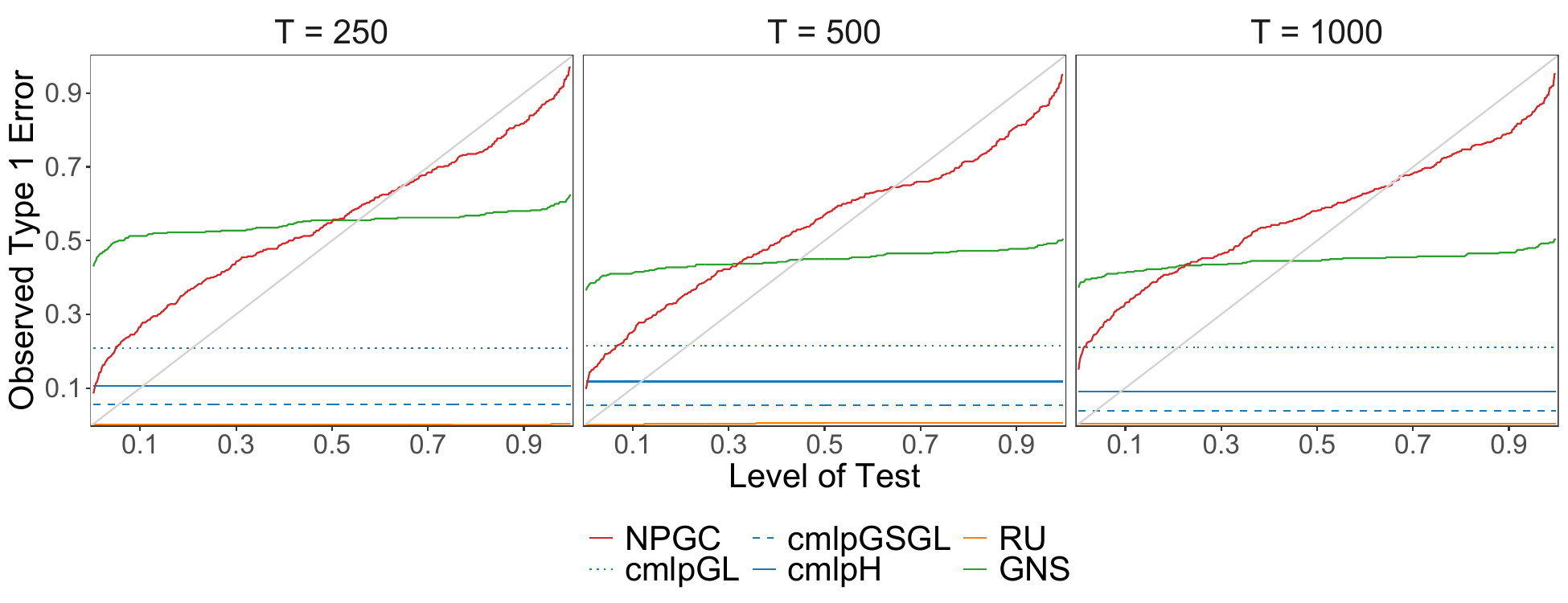}
			\caption{Type 1 error control for Lorenz-96 simulations. Methods that adhere to a chosen level of test will lie closer to the uniform CDF (gray) included for clarity that indicates the reference quantile $\alpha$.}
		\label{figure1b}
	\end{figure}
	The NPGC method succeeds in identifying many cases of Granger causal influence in both processes, and as expected, this ability approaches the upper limit as the number of time points increases. The cMLP methods provide spotty results, and these are greatly influenced by the selected parameter $\lambda$. Without guidance for a specific penalty selection, the global decision of ``Granger causal'' or ``non-causal'' is uncertain. The naive restricted and unrestricted model method fails to identify many Granger causal pairings in the Lorenz-96 trials. The Gaussian noise substitution method correctly identifies nearly all Granger causal pairs, but this comes at the cost of uncontrolled Type 1 error (see Figures \ref{figure1a} and \ref{figure1b}). 
	\par
	Under the null hypothesis, moderate adherence to the asymptotic properties shown in Theorem \ref{theorem3} is demonstrated. The NPGC permutation methodology tracks relatively close to the included reference CDFs in Figure \ref{figure1a} for TAR(2) simulation settings, and exhibits mild to moderate deviations for Lorenz-96 settings in Figure \ref{figure1b}. These deviations can likely be attributed to the complicated nature of the chaotic Lorenz system; the same pattern is not observed in other simulations when both groups are generated by a TAR(2) process (see additional figures in the supplement). The cMLP lasso methods are included as horizontal lines in these plots; they do not have a straightforward extension for a specific level of test, other than computationally expensive iteration over several penalties $\lambda$. The naive methods do not adhere to their theoretical Type 1 error control.
	\par
	Granger causal identification methods that use a penalized objective can be useful, but do not allow for the global decision to label a group of variables ``Granger causal'' or ``non-causal''. These methods should be applied to narrow the scope of research to individual variables of the covariate set after a global method is applied. 
	
	\section{Application Study}\label{se:NPGCapplication}
	The out-of-sample permutation framework is applied to neuronal responses to acoustic stimuli in the primary auditory cortex of an anesthetized (ketamine–medetomidine) rat. The data, taken from the Collaborative Research in Computational Neuroscience data sharing website, consists of in vivo whole-cell recordings (mV) sampled at 4kHz in response to natural sound fragments \citep{machens04, asari09}. A time region of interest between 11 and 15 seconds is isolated from experimental trial 60 of the 050802mw03 data (partially shown in Figures 2D and 2E of \citet{machens04}). The subset of data corresponds to the recordings in response to a jaguar (\textit{Panthera onca}) mating call sound fragment presented at 97.656 kHz; this is resampled to a frequency of 4kHz matching that of the response. Additional sound fragments of a Humpback whale (\textit{Megaptera novaeangliae}) and Knudsen's frog (\textit{Leptodactylus knudseni}), played in other trials throughout the experiment, are included to construct a simplistic, non-causal scenario. Figure \ref{figureapp1} displays the whole-cell recordings to the jaguar mating call and all sound fragments examined. \citet{machens04} contains additional detail on data collection and experimental methods.
	\begin{figure}[htb]
			\centering
			\includegraphics[width = \textwidth]{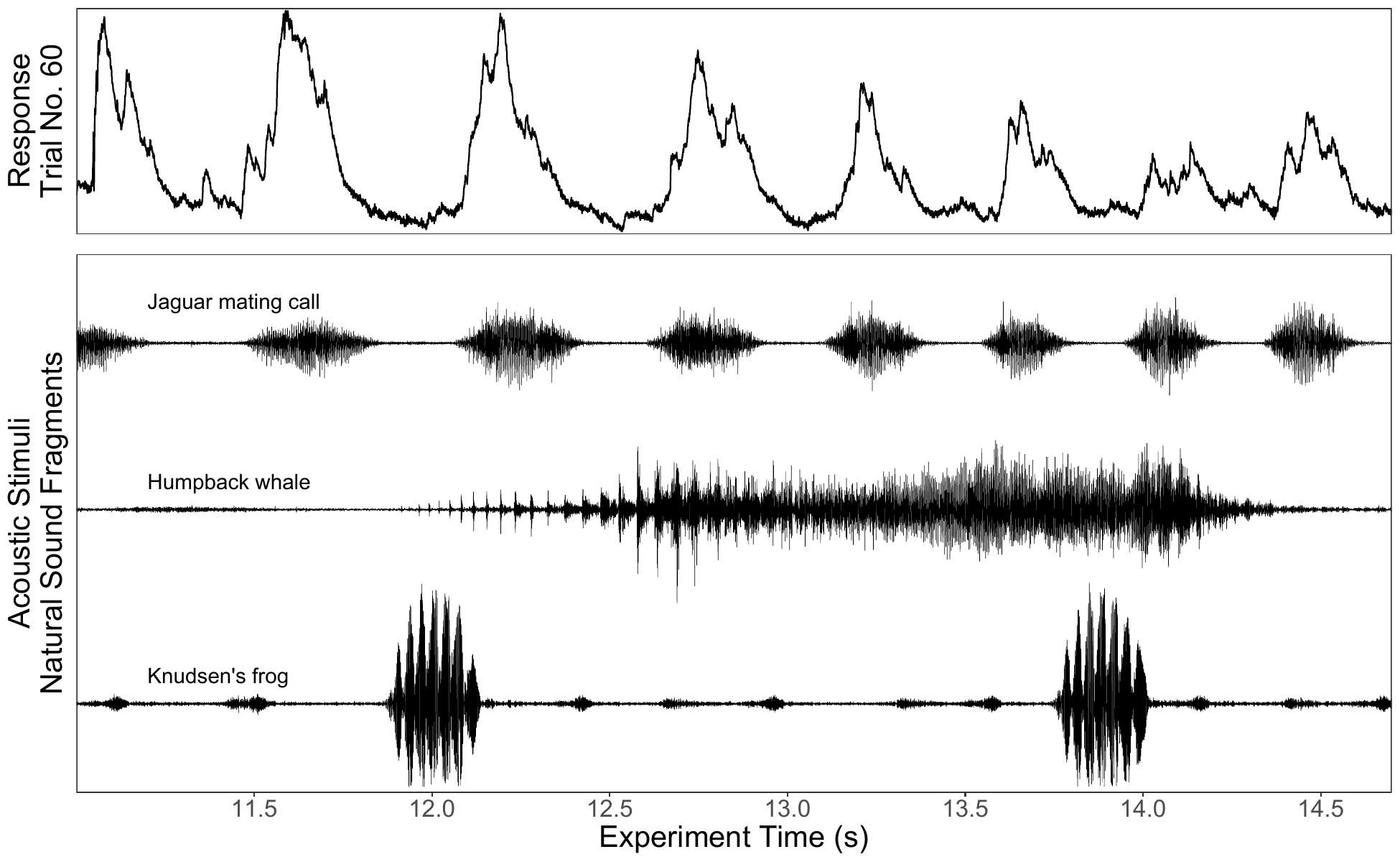}
			\caption{\textit{Top}: In vivo whole-cell recordings (mV) from the contralateral (left) primary auditory cortex of an anesthetized rat. \textit{Bottom}: Jaguar mating call acoustic stimulus fragment resampled at 4kHz. Humpback whale and Knudsen's frog acoustic stimuli are included as non-causal examples for the methodology.}
		\label{figureapp1}
	\end{figure}
	\par
	Performance of the Granger causal detection methodology is evaluated by formulating a farcical example in which any suitable method should be able to recover the correct causal structure. The methods are expected to flag the jaguar sound fragment as causal, while correctly labelling the other two included fragments non-causal. Lagged information from 30ms to 40ms is isolated for inclusion, consistent with the measured the half-maximal synaptic conductance in \citet{wehr03}. One acoustic stimulus fragment is chosen as the covariate set $\mathbf{X}$ and the others are additional variables $\mathbf{Z}$ in the model. Lagged values of the response are included up to 10ms. The NPGC method is implemented with $M=400$ permutations, $K=5$ cross validation folds, $\mathscr{R}=50$ randomly generated FNNs, and a feature dimension of $N=250$ due to the large number of covariates. The comparator penalized optimization methods are examined at the same dimension and a variety of penalty values. The selection of the penalty $\lambda$ presents a major challenge, and this value must be chosen prior to analysis; it should not be fine-tuned after the fact to produce a desired (already known) result. Results for the fabricated, naive scenario are compiled in Figure \ref{figureapp2}.
	\begin{figure}[htb]
			\centering
			\includegraphics[width = \textwidth]{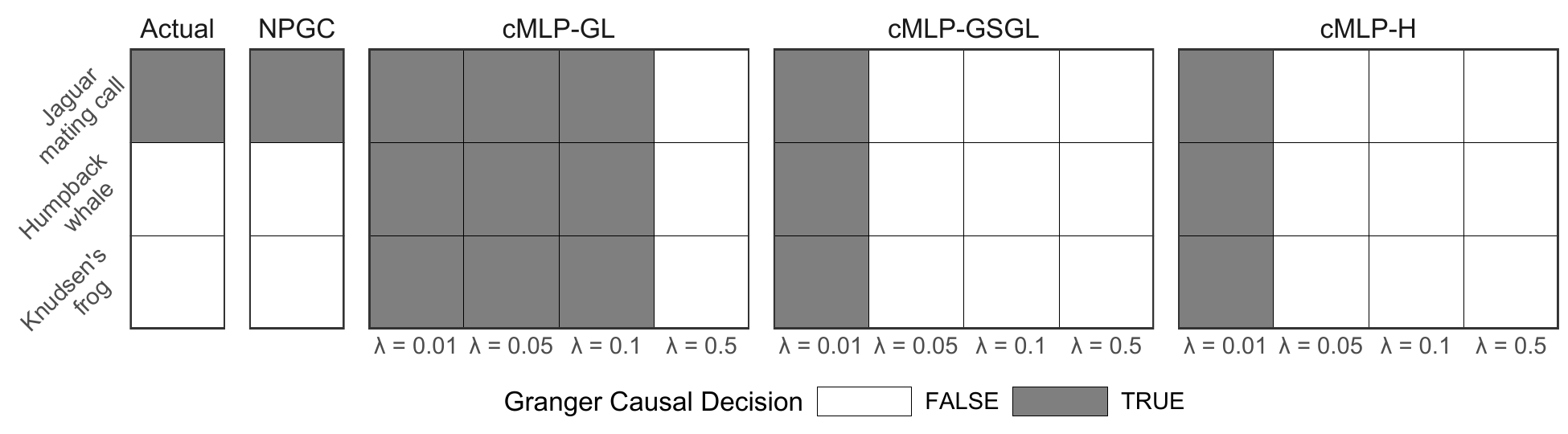}
			\caption{Estimated Granger causal relationship between each acoustic stimulus and primary auditory cortex response. NPGC quantiles $\hat{Q}_M = 0.0025, 1.000, 0.4425,$ from top to bottom. The NPGC method recovers the correct causal structure, while the lasso-based cMLP methods fail to isolate the jaguar stimulus.}
		\label{figureapp2}
	\end{figure}
	Permutation based methodology extracts the correct relationship between the jaguar mating call sound fragment and the whole-cell recording response, and the penalized variable selection approaches do not.
	
	\section{Discussion}\label{se:NPGCdiscussion}
	The NPGC methodology illustrates the use of the permutation-based shift from in-sample to out-of-sample testing. Permutation tests are used widely in literature, and \citet{nauta19} implement a decision rule similar to NPGC. This manuscript explicitly defines Granger causality in this framework, and it provides a theoretical analysis of the corresponding variance estimates and decision rule. The shift allows for control in identifying useful functional relationships when overfitting from an artificial network becomes a concern, and removes the burden of specifying regularization parameters that have a major impact on the observed outcome. Once a Granger causal outcome has been determined, penalized variable selection techniques provide a practical screening method for identifying the prominent relationships, but they do not effectively estimate the global presence or absence of functional connectivity.
	\par
	The permutation framework is able to detect the presence of functional connectivity, and provide a safeguard against misidentification of data-specific noise as functional dependence when the artificial network is overfit. Alternative formulations of the featurizing function $\Psi$ in Equation \ref{featurizer} can utilize the same methodological structure and retain the theoretical guarantees provided the transformation meets the listed conditions and the construction allows for row-wise permutation without breaking the dependence structure across variables in $\mathbf{X}$. Misspecification of the dimension of the feature space should only produce additional false negative results as long as $N$ is less than the number of data points used in model estimation. Using a large enough dimension to capture any potential nonlinear dependence structure that may exist in the dataset is suggested. The NPGC method may exhibit slight undercoverage for a chosen level of test, but \emph{minor} deviations in this realm are not of great concern if the result is correctly interpreted as a potential functional connection and not an outright causal effect. The permutation method circumvents the need for a penalty parameter selection, and provides ease of extension to multiple testing problems and control of family-wise error rate. 
	\par
	Potential misuses of Granger causality include, but are not limited to, repeated application to subsets of a selected covariate group without adequate Type 1 error correction, disregard for the conditional nature of the inferential conclusion, attempts at individual rather than collective covariate inference (unless model specification is complete and exact), inclusion of a covariate without careful scientific or logical reasoning, and use as an outright mechanism for identifying causal relationships rather than predictive links for future study. Prudent selection of the length of lag response included in the covariate matrix is required, and two values may produce different results. In nonlinear processes, selection of the relevant history of the response for inclusion in the model remains an area of future study.
	
	\bibliography{NPGC_arXiv}
	
	\newpage
	\begin{center}
		{\large\bf SUPPLEMENTARY MATERIAL}\\
		{\normalsize \bf Title}
	\end{center}
	\begin{description}
		\item[Additional Tables \& Figures:] Tables and figures of simulation and application results not presented in the main text.
		\item[Algorithms \& Pseudocode:] Algorithms and pseudocode for methodology presented in the main text.
		\item[Theory Supplement:] Additional conditions for theoretical results, finite sample behavior, and proofs to theorems presented in the main text.
		\item[Code \& Supporting Material:] Files (.RData, .npy, and .csv) used to assess performance of change point methods, code (.R, .cpp, and .py) used to generate results and figures, and data (.mat) and code (.R, .py, and .m) used to generate output in Section \ref{se:NPGCapplication} can be found at $$\verb|github.com/noahgade/NonlinearPermutedGrangerCausality|.$$
		
	\end{description}
	\pagebreak
	\FloatBarrier
	\section*{Additional Tables \& Figures}
	\subsection*{Tables}
	\begin{table}[htb]
		\centering
			\begin{tabular}{ccccccc}\hline
				& NPGC & cMLP-GL & cMLP-GSGL & cMLP-H & R/U & GNS \\\hline
				$T=250$ & 0.895 & 0.795 & 0.110 & 0.435 & 0.865 & 0.985 \\
				$T=500$ & 0.975 & 0.805 & 0.130 & 0.505 & 0.910 & 1.000 \\
				$T=1000$ & 0.990 & 0.805 & 0.090 & 0.420 & 0.915 & 1.000 \\\hline
			\end{tabular}
			\caption{TAR(2) simulation results. Proportion of correctly labelled Granger causal outcomes ($\text{GC} = 1$, $\rho_1$ in Table \ref{table2}) for zero Granger causal variables included in additional set $\mathbf{Z}$.}
			\label{tableA}
		\end{table}
		\begin{table}[htb]
			\centering
			\begin{tabular}{ccccccc}\hline
				& NPGC & cMLP-GL & cMLP-GSGL & cMLP-H & R/U & GNS \\\hline
				$T=250$ & 0.935 & 0.945 & 0.200 & 0.675 & 0.640 & 0.990 \\
				$T=500$ & 0.955 & 0.900 & 0.135 & 0.610 & 0.775 & 0.995 \\
				$T=1000$ & 0.980 & 0.895 & 0.155 & 0.580 & 0.785 & 1.000 \\\hline
			\end{tabular}
			\caption{TAR(2) simulation results. Proportion of correctly labelled Granger causal outcomes ($\text{GC} = 1$, $\rho_1$ in Table \ref{table2}) for two Granger causal variables included in additional set $\mathbf{Z}$.}
			\label{tableB}
	\end{table}
	
	\begin{table}[htb]
			\centering
			\begin{tabular}{ccccccc}\hline
				& NPGC & cMLP-GL & cMLP-GSGL & cMLP-H & R/U & GNS \\\hline
				$T=250$ & 0.955 & 0.705 & 0.220 & 0.435 & 0.050 & 1.000 \\
				$T=500$ & 0.980 & 0.650 & 0.115 & 0.380 & 0.015 & 1.000 \\
				$T=1000$ & 0.985 & 0.670 & 0.090 & 0.360 & 0.040 & 1.000 \\\hline
			\end{tabular}
			\caption{Lorenz-96 simulation results. Proportion of correctly labelled Granger causal outcomes ($\text{GC} = 1$, $\rho_1$ in Table \ref{table2}) for zero Granger causal variables included in additional set $\mathbf{Z}$.}
			\label{tableC}
		\end{table}
		\begin{table}[htb]
			\centering
			\begin{tabular}{ccccccc}\hline
				& NPGC & cMLP-GL & cMLP-GSGL & cMLP-H & R/U & GNS \\\hline
				$T=250$ & 1.000 & 0.825 & 0.245 & 0.530 & 0.000 & 1.000 \\
				$T=500$ & 0.995 & 0.805 & 0.205 & 0.505 & 0.000 & 1.000 \\
				$T=1000$ & 1.000 & 0.760 & 0.165 & 0.465 & 0.000 & 1.000 \\\hline
			\end{tabular}
			\caption{Lorenz-96 simulation results. Proportion of correctly labelled Granger causal outcomes ($\text{GC} = 1$, $\rho_1$ in Table \ref{table2}) for two Granger causal variables included in additional set $\mathbf{Z}$.}
			\label{tableD}
	\end{table}
	\FloatBarrier
	
	\subsection*{Figures}
	\begin{figure}[htb]
			\centering
			\includegraphics[width = \textwidth]{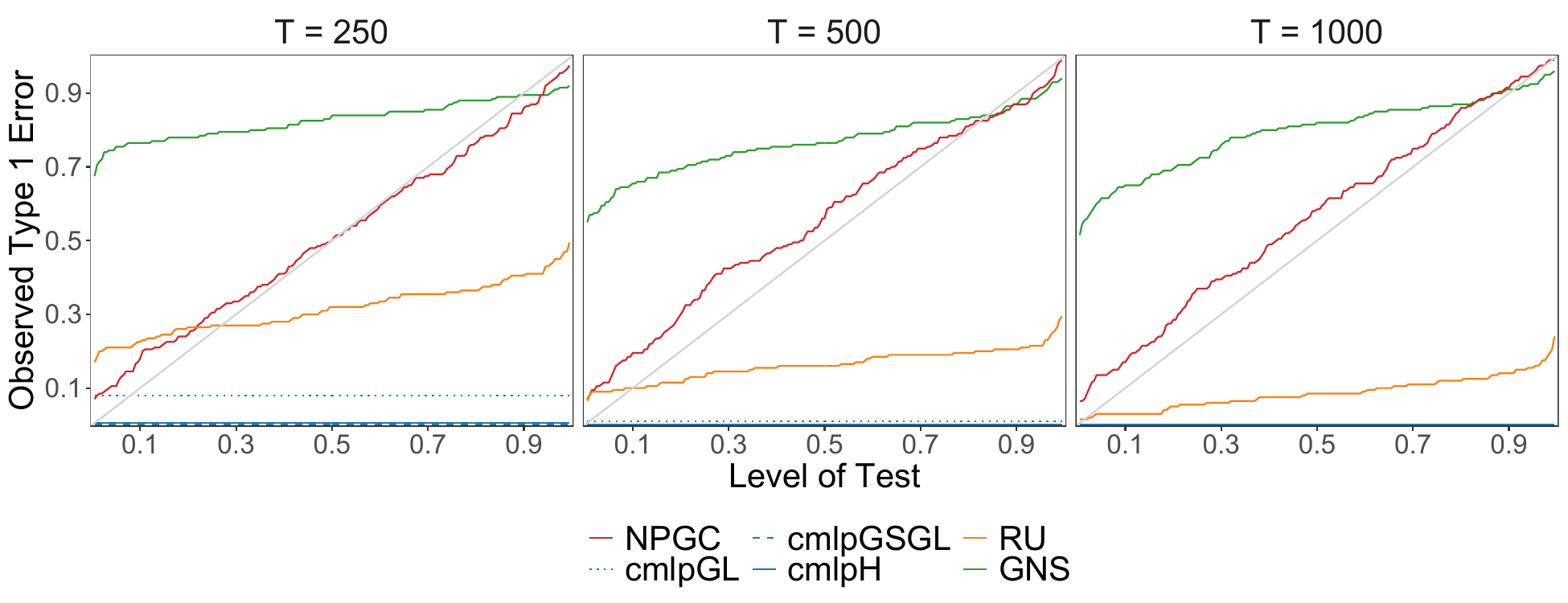}
			\caption{Type 1 error control for TAR(2) simulations by reference quantile $\alpha$ with zero Granger causal variables included in additional set $\mathbf{Z}$.}
			\label{figureA}
		\end{figure}

		\begin{figure}[htb]
			\centering
			\includegraphics[width = \textwidth]{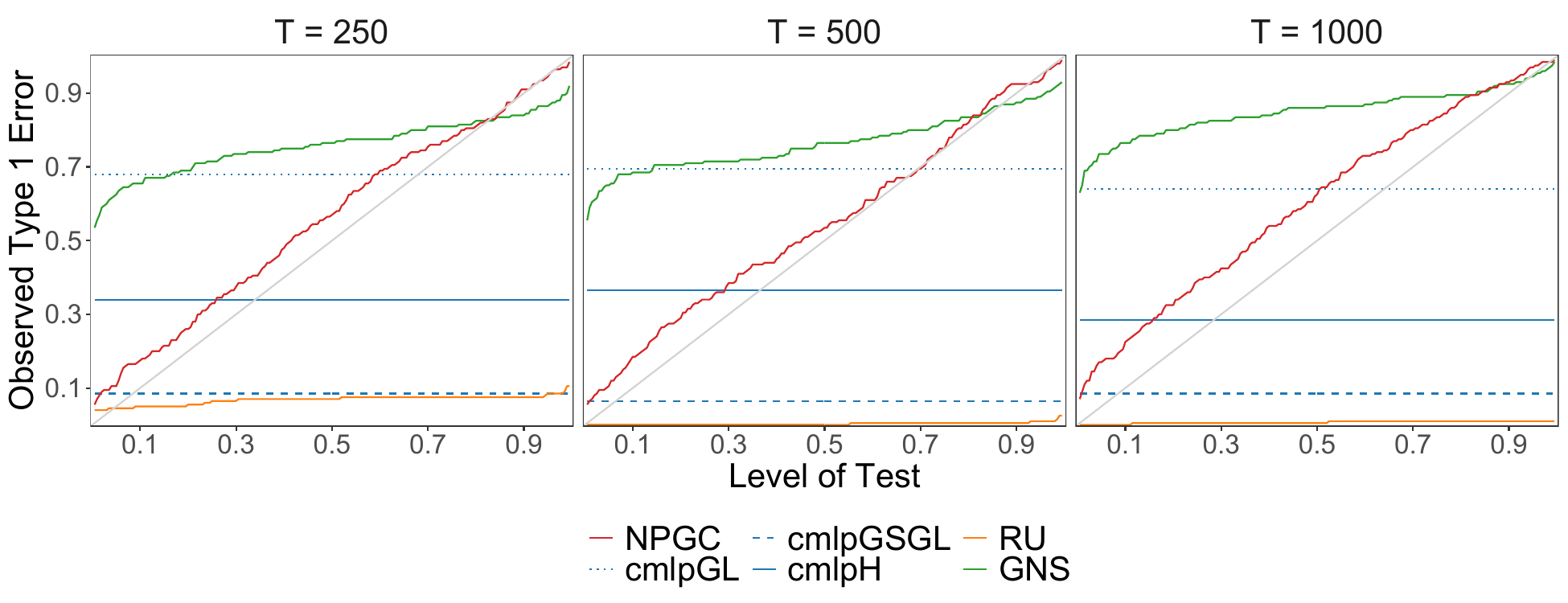}
			\caption{Type 1 error control for TAR(2) simulations by reference quantile $\alpha$ with two Granger causal variables included in additional set $\mathbf{Z}$.}
			\label{figureB}
	\end{figure}
	
	\begin{figure}[htb]
			\centering
			\includegraphics[width = \textwidth]{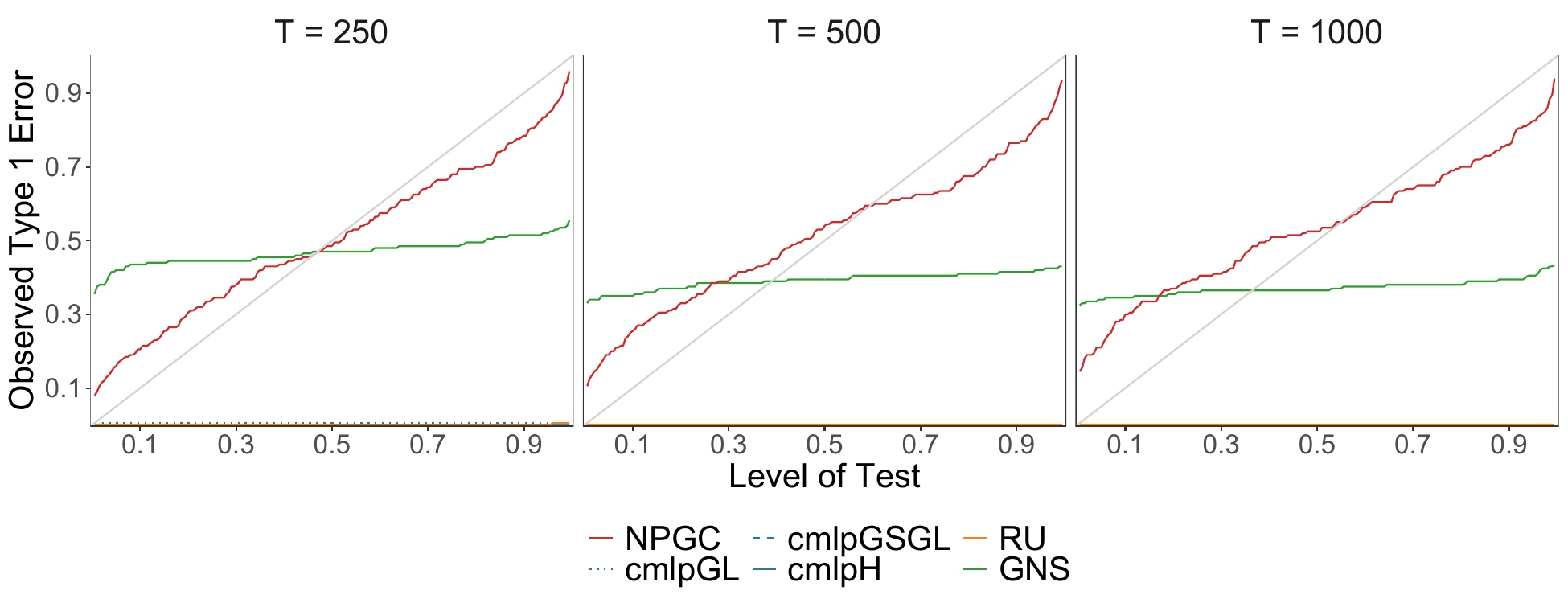}
			\caption{Type 1 error control for Lorenz-96 simulations by reference quantile $\alpha$ with zero Granger causal variables included in additional set $\mathbf{Z}$.}
			\label{figureC}
		\end{figure}

			\begin{figure}[htb]
			\centering
			\includegraphics[width = \textwidth]{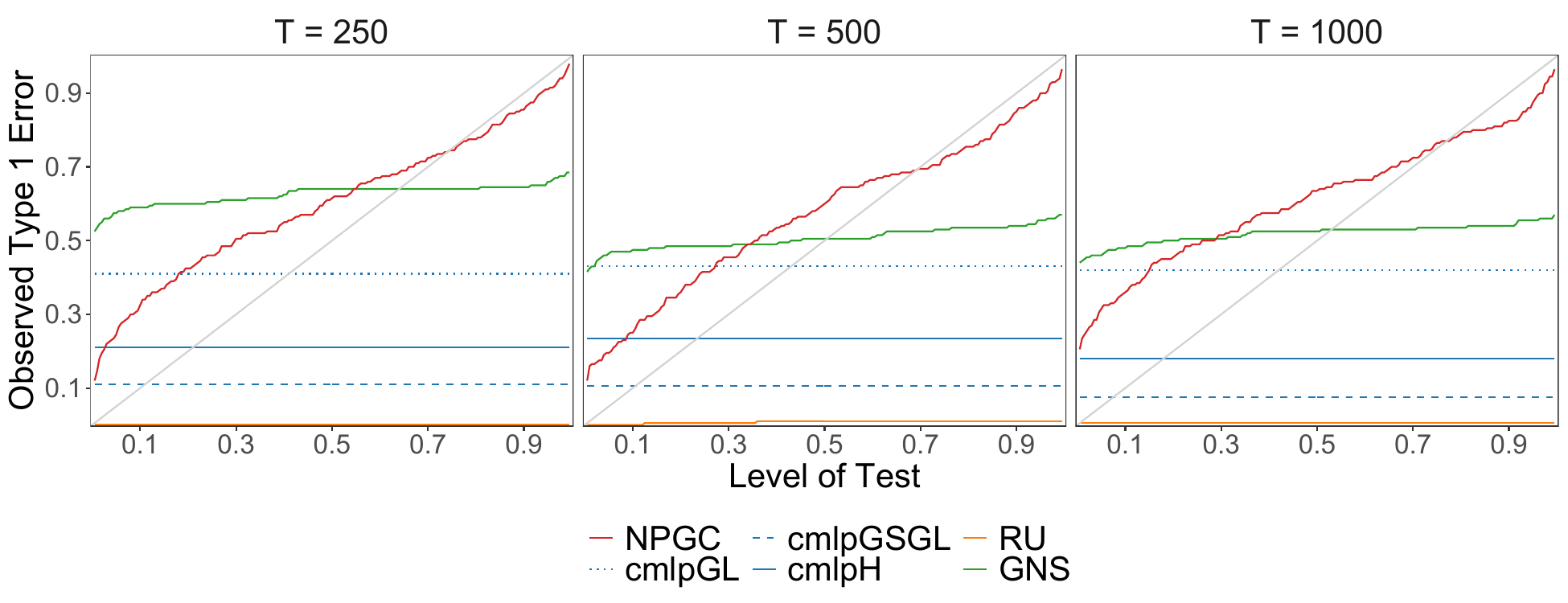}
			\caption{Type 1 error control for Lorenz-96 simulations by reference quantile $\alpha$ with two Granger causal variables included in additional set $\mathbf{Z}$.}
			\label{figureD}
	\end{figure}
	
	\begin{figure}[htb]
		\centering
			\includegraphics[width = \textwidth]{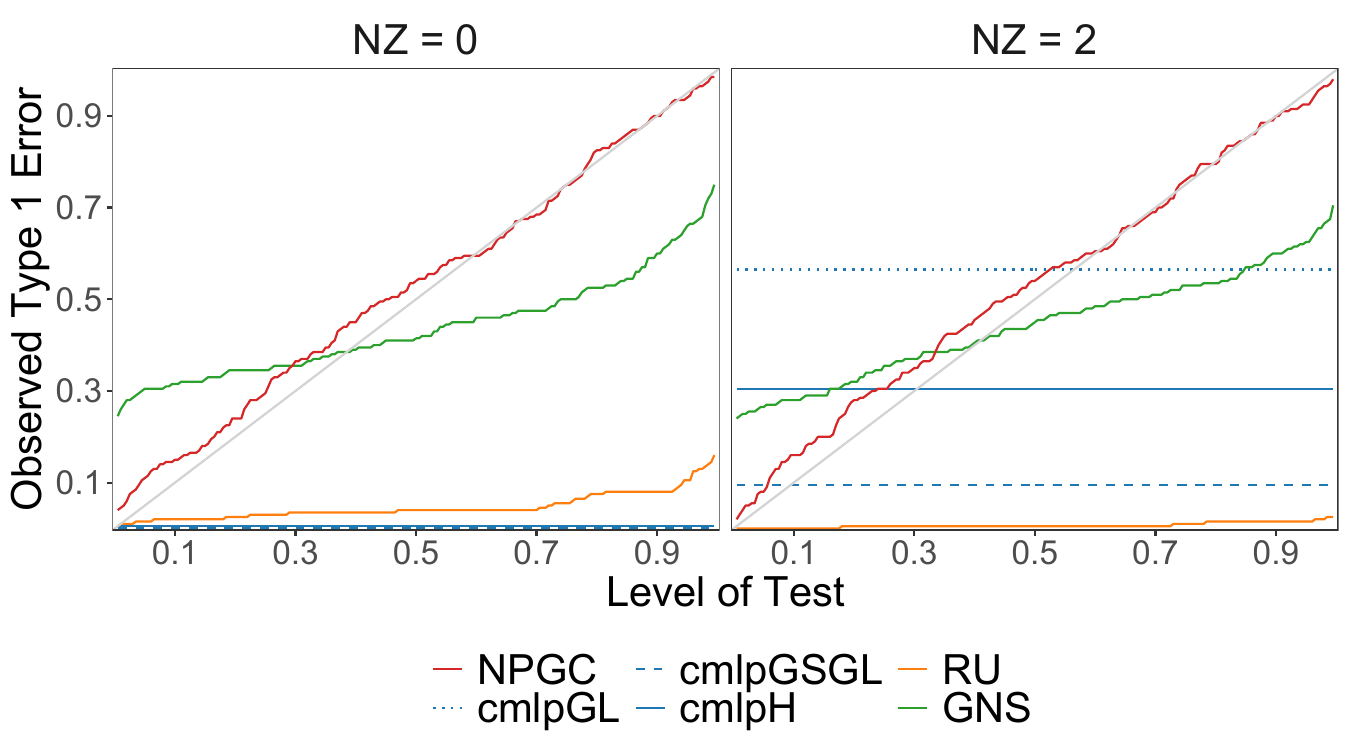}
			\caption*{Type 1 error control for two-group TAR(2) simulations by reference quantile $\alpha$ simulations without Lorenz-96 group included. \textit{Left}: Zero Granger causal variables. \textit{Right}: Two Granger causal variables included in the additional set $\mathbf{Z}$. All data sets are $T=1000$.}
		\label{figureE}
	\end{figure}
	
	\pagebreak
	\FloatBarrier
	\section*{Algorithms \& Pseudocode}
	
	\floatname{algorithm}{Algorithm}
	\begin{algorithm}
		\caption{Nonlinear Permuted Granger Causality}
		\label{npgcalg}
		\vskip0.05in
		\textbf{Inputs}: $\left(\mathbf{X}, \mathbf{Y}, \mathbf{Z}\right)_\omega$ for all $\varphi$ realizations $\omega\in\Omega_\text{obs}$; lag selection $\gamma$; \# permutations $M$; \# random featurizations $\mathscr{R}$; feature space dimension $N$; \# cross-validation folds $K$
		\vskip0in
		\textbf{Outputs}: $\hat{Q}_M$; $\hat{\vartheta}_{m}$ for each permutation $m=1,\ldots,M$
		\vskip0in
		\begin{algorithmic}[1]
			\State Generate permutations $\tilde{\mathbf{X}}_m = \bm{\Pi}_m\mathbf{X}$ for $m=1,\ldots,M$ with $\bm{\Pi}_1 = \mathbf{I}$
			\State Initialize $\mathbf{W}_r\in\mathbb{R}^{(1 + \gamma d + q + p)\times N}$ where each element $w_{r,ij} \sim \mathcal{N}\left(0, 1\right)$ for all $\mathscr{R}$
			\For{$m$ in $1:M$}
			\For{$\omega$ in $1:\varphi$}
			\For{$r$ in $1:\mathscr{R}$}
			\State $\mathbf{H}_{m,\omega, r} \leftarrow \tanh\left(\left[\mathbf{1}\;\mathbf{Y}_{\text{lag},\omega}\;\mathbf{Z}_{\omega}\;\tilde{\mathbf{X}}_{m,\omega}\right] \mathbf{W}_r\right)$
			\For{$k$ in $1:K$}
			\State $\mathbf{R}_{m,\omega,r,k} \leftarrow \mathbf{H}_{m,\omega,r,k} \left(\mathbf{H}_{m,\omega,r,-k}'\mathbf{H}_{m,\omega,r,-k}\right)^{-1}\mathbf{H}_{m,\omega,r,-k}'\mathbf{Y}_{\omega,-k} - \mathbf{Y}_{\omega,k}$
			\EndFor
			\EndFor
			\EndFor
			\State $\hat{\vartheta}_{m} \leftarrow \left(\varphi \mathscr{R}K \right)^{-1} \sum_{\omega=1}^{\varphi} \sum_{r=1}^{\mathscr{R}}\sum_{k=1}^K T_k^{-1} \text{tr}\left(\mathbf{R}_{m,\omega,r,k}'\mathbf{R}_{m,\omega,r,k}\right)$
			\EndFor
			\State $\hat{Q}_{M} \leftarrow M^{-1}\sum_{m=1}^M\mathbf{1}\left\{\hat{\vartheta}_{m} \leq \hat{\vartheta}_{1}\right\}$
			\vskip0in
			\Return $\hat{Q}_{M}$; $\hat{\vartheta}_{m}$ for $m=1,\ldots,M$
		\end{algorithmic}
	\end{algorithm}
	\pagebreak
	Automated selection of the feature space dimension in Algorithm \ref{fdselectalg} uses cross validation sets $1,\ldots,k^*$ of the $K$ total to obtain $N$. Implementation of Algorithm \ref{fdselectalg} requires alteration of Algorithm \ref{npgcalg} to only use test sets $k = k^* + 1,\ldots,K$.
	\floatname{algorithm}{Algorithm}
	\begin{algorithm}
		\caption{Automated Feature Space Dimension Selection }
		\label{fdselectalg}
		\vskip0in
		\textbf{Inputs}: $\left(\mathbf{X}, \mathbf{Y}, \mathbf{Z}\right)_\omega$ for all $\varphi$ realizations $\omega\in\Omega_\text{obs}$; lag selection $\gamma$; \# random FNN generations $\mathscr{R}$; \# cross-validation folds $K$
		\vskip0in
		\textbf{Outputs}: feature space dimension $N$
		\vskip0in
		\begin{algorithmic}[1]
			\State Remove set $k^* + 1,\ldots,K$ from $\left(\mathbf{X}, \mathbf{Y}, \mathbf{Z}\right)_\omega$ leaving only folds $1,\ldots,k^*$
			\State $N_{\max} \leftarrow \sum_{k=1}^{k^*}T_k-1$ 
			\For{$n$ in $10:10:N_{\max}$}
			\State Initialize $\mathbf{W}_r\in\mathbb{R}^{\left(1 + \gamma d + q + p\right)\times n}$ where each element $w_{r,ij} \sim \mathcal{N}\left(0, 1\right)$ for all $\mathscr{R}$
			\For{$\omega$ in $1:\varphi$}
			\For{$r$ in $1:\mathscr{R}$}
			\State $\mathbf{H}_{n,\omega,r} \leftarrow \tanh\left(\left[\mathbf{1}\;\mathbf{Y}_{\text{lag}, \omega}\;\mathbf{Z}_{\omega}\;\mathbf{X}_{\omega}\right] \mathbf{W}_r\right)$
			\For{$k$ in $1:k^*$}
			\State $\mathbf{R}_{n,\omega,r,k} \leftarrow \mathbf{H}_{n,\omega,r,k} \left(\mathbf{H}_{n,
				\omega,r,-k}'\mathbf{H}_{n,\omega,r,-k}\right)^{-1}\mathbf{H}_{n,\omega,r,-k}'\mathbf{Y}_{\omega,-k} - \mathbf{Y}_{\omega,k}$
			\EndFor
			\EndFor
			\EndFor
			\EndFor
			\State $N \leftarrow \arg \min_{n}\left\{\left(\varphi \mathscr{R} k^*\right)^{-1}\sum_{\omega=1}^{\varphi}\sum_{r=1}^\mathscr{R}\sum_{k=1}^{k^*} T_k^{-1}\text{tr}\left(\mathbf{R}_{n,\omega,r,k}'\mathbf{R}_{n,\omega,r,k}\right)\right\}$
			\vskip0in
			\Return $N$
		\end{algorithmic}
	\end{algorithm}
	
	\FloatBarrier
	\newpage
	\FloatBarrier
	\section*{Theory Supplement}
	\subsection*{Additional Conditions for Theoretical Results}
	Theoretical results are derived for a generic activation function $g$ in Equation \ref{featurizer} with the constraints of Condition \ref{functioncondition}.
	\begin{condition}\label{functioncondition}
		The nonlinear function activation function $g$ in Equation \ref{featurizer} is bounded, $g:\mathbb{R}\rightarrow \left[a,b\right]$ for some $a<b$, such that for any $x\in\mathbb{R}$, $|g(x)|\leq G=\max\left\{|a|, |b|\right\}$ and $G<\infty$.
	\end{condition}
	\begin{condition}\label{rankcondition}
		Let $\mathcal{W}$ be the space of all element-wise randomly generated FNNs, $w_{ij}\sim\mathcal{N}(0,1)$, such that for all $\mathbf{W}_r\in\mathcal{W}$, the matrix $\mathbf{W}_r$ is full rank and generates a feature matrix that is full rank with a finite condition number. For permutation $m$, realization $\omega$, and random featurization $r$,
		\begin{align}
			\text{rank}\left(\mathbf{H}_{m,\omega,r}\right)&=N\label{rankconditionequation}\\
			\text{and}\hskip0.2in\kappa(\mathbf{H}_{m,\omega,r})&=\sigma_1(\mathbf{H}_{m,\omega,r})/\sigma_{N}(\mathbf{H}_{m,\omega,r})\leq \kappa_{\max} <\infty.\label{conditionequation}
		\end{align}
		All initialized model matrices $\mathbf{W}_r\in\mathcal{W}$.
	\end{condition}
	\par
	For the matrix $\mathbf{H}_{m,\omega,r}=(h_{m,\omega,r, ij})$, define $G$ as the maximal element from Condition \ref{functioncondition} and the average squared entry as $\bar{h}_{m,\omega,r}^2$.
	\begin{condition}\label{entriescondition}
		The following expected values exist and are finite:
		\begin{align}
			\mathbb{E}\left[G^{-2}\bar{h}_{m,\omega,r}^2|\mathbf{S}_{m,\omega,r}\right] &= \nu^2<\infty,\\
			\mathbb{E}\left[G^2\left(\bar{h}_{m,\omega,r}^2\right)^{-1}|\mathbf{S}_{m,\omega,r}\right] &= \xi^2<\infty,\\
			\text{and}\hskip0.1in\mathbb{E}\left[G^4\left(\bar{h}_{m,\omega,r}^2\right)^{-2}|\mathbf{S}_{m,\omega,r}\right] &= \varrho^4<\infty.
		\end{align}
	\end{condition}
	\par
	Combining Conditions \ref{functioncondition}, \ref{rankcondition}, and \ref{entriescondition}, $0<\nu^2\leq 1$, $1\leq \xi^2<\infty$, and $1\leq \varrho^4 < \infty$. Define $f$ as the true functional relationship between the response and the unpermuted predictor matrix for all $\omega\in\Omega$, 
	\begin{align}
		\mathbf{Y}_\omega=f\left(\left[\mathbf{1}\;\mathbf{Y}_{\text{lag},\omega}\;\mathbf{Z}_\omega\;\tilde{\mathbf{X}}_{1,\omega}\right]\right) + \bm{\mathscr{U}}_{\omega},\label{truef}
	\end{align}
	and the approximating model form for permutations $m=1,\ldots,M$ and featurizations $r=1,\ldots,\mathscr{R}$.
	\begin{align}\label{modelequation}
		\mathbf{Y}_\omega = \mathbf{H}_{m,\omega,r}\mathbf{W}^L_{m,\omega,r} + \mathbf{U}_{m,\omega,r}
	\end{align}
	\begin{condition}\label{approximationcondition}
		For every $\eta>0$, there exists a fixed $N$, where $1 + \gamma d + q + p\leq N<\infty$, such that as the number of random featurizations $\mathscr{R}\rightarrow\infty$,
		\begin{align}\label{approximationconditionequation}
			\sup_{\omega\in\Omega} \left\| \frac{1}{\mathscr{R}}\sum_{r=1}^{\mathscr{R}}\mathbf{H}_{1,\omega,r}\mathbf{W}^L_{1,\omega,r} - f\left(\left[\mathbf{1}\;\mathbf{Y}_{\text{lag},\omega}\;\mathbf{Z}_\omega\;\tilde{\mathbf{X}}_{1,\omega}\right]\right)\right\| < \eta,
		\end{align}
		where $\mathbf{W}^L_{1,\omega,r}$ is the true coefficient matrix of feature space $r$ for the unpermuted covariate set, $\mathbf{H}_{1,\omega,r} = g\left(\left[\mathbf{1}\;\mathbf{Y}_{\text{lag},\omega}\;\mathbf{Z}_\omega\;\tilde{\mathbf{X}}_{1,\omega}\right]\mathbf{W}_r\right)$.
	\end{condition}
	\par
	Note that Condition \ref{approximationcondition} specifically pertains to the $f$ piece of the true functional relationship; no assumption is made of the closeness of a transformed response $\mathbf{H}_{1,\omega,r}\mathbf{W}^L_{1,\omega,r}$ and the true values $\mathbf{Y}_\omega$ if the predictors themselves are not reliable and the entries of $\bm{\mathscr{U}}_{\omega}$ in Equation \ref{truef} are large. 
	
	\subsection*{Finite Sample Distribution}\label{ss:FiniteNPGC}
	For a finite sample, the distribution of the estimated quantity $\hat{\vartheta}_{m}$ is derived as a sum of linear combinations of chi-square random variables. Define the following chi-square random variables 
	\begin{align}
		X_{m,\omega,r,k,i},Y_{m,\omega,r,k,ij}, Z_{m,\omega,r,k,ij} \sim \chi_1^2
	\end{align}
	for all $\omega\in\Omega_\text{obs}$, $r=1,\ldots,\mathscr{R}$, $k=1,\ldots,K$, $i=1,\ldots,T_kd$ and $j<i$.  Denote 
	\begin{align}
		\mathbf{H}_{m,\omega,r,k}\left[\mathbf{H}_{m,\omega,r,-k}'\mathbf{H}_{m,\omega,r,-k}\right]^{-1}\mathbf{H}_{m,\omega,r,k}' = (\phi_{m,\omega,r,k,ij})
	\end{align}
	and $\mathbf{S}_{m,\omega,r}=(s_{m,\omega,r,ij})$, with $i'=\lceil i/d\rceil$,  $j'=\lceil j/d\rceil$, $i^*=i \bmod d$, and $j^*=j \bmod d$. 
	\begin{theorem}\label{theorem5}
		Under the listed conditions, a finite sample containing $T$ observations, $\mathscr{R}$ random generated FNNs, and $\varphi$ realizations in the set $\Omega_\text{obs}$, the estimate for the variation parameter $\hat{\vartheta}_{m}$ defined in Equation \ref{parameter} follows a generalized chi-square distribution.
		\begin{align}
			\hat{\vartheta}_{m} &\sim \frac{1}{\varphi\mathscr{R}K}\sum_{\omega = 1}^{\varphi} \sum_{r=1}^{\mathscr{R}} \sum_{k=1}^K \frac{1}{T_k} \sum_{i=1}^{T_kd} \bigg[\left(\phi_{m,\omega,r,k,i'i'}+1\right)s_{m,\omega,r,k,i^*i^*}X_{m,\omega,r,k,i} \nonumber\\
			&\hskip0.5in + \sum_{j=1}^{d\lfloor (i-1) / d \rfloor} \phi_{m,\omega,r,k,i'j'}s_{m,\omega,r,k,i^*j^*}\left(Y_{m,\omega,r,k,ij} - Z_{m,\omega,r,k,ij}\right) \nonumber\\
			&\hskip0.5in + \sum_{j=d\lfloor (i-1) / d\rfloor + 1}^{i-1}\left(\phi_{m,\omega,r,k,i'j'} + 1\right)s_{m,\omega,r,k,i^*j^*}\left(Y_{m,\omega,r,j,ij} - Z_{m,\omega,r,k,ij}\right) \bigg]
		\end{align}
	\end{theorem}
	
	\subsection*{Proofs}
	\begin{lemma}\label{lemma1}
		Under the listed conditions, $\lim_{T\rightarrow\infty} \mathbb{E}\left[\hat{\vartheta}_m\right] = \vartheta_m$.
	\end{lemma}
	\par
	The proof of Lemma \ref{lemma1} begins with the estimate for one test set with a single individual randomly generated FNN given the variation measure $\mathbf{S}_{m,\omega,r}$. We aggregate these expectations over the sets $k=1,\ldots,K$, generated FNNs $r=1,\ldots,\mathscr{R}$ and the $\varphi\geq1$ observed datasets in $\Omega_\text{obs}$.
	\begin{proof}[Proof of Lemma \ref{lemma1}]
		We seek the conditional expectation $\mathbb{E}\left[\text{tr}\left(\mathbf{R}_{m,\omega,r,k}'\mathbf{R}_{m,\omega,r,k}\right)|\mathbf{S}_{m,\omega,r}\right]$, and can write the unconditional expectation of the underlying parameter based on the definition in Equation \ref{parameter}.
		\begin{align}
			\mathbb{E}\left[\hat{\vartheta}_m\right] &= \frac{1}{\varphi\mathscr{R}K}\sum_{\omega=1}^{\varphi}\sum_{r=1}^{\mathscr{R}}\sum_{k=1}^K\frac{1}{T_k}\mathbb{E}\left[\mathbb{E}\left(\mathbb{E}\left[\text{tr}\left(\mathbf{R}_{m,\omega,r,k}'\mathbf{R}_{m,\omega,r,k}\right)|\mathbf{S}_{m,\omega,r}\right]|\bm{\Sigma}_{m,\omega}\right)\right]\label{fullexp1}
		\end{align}
		From Condition \ref{modelerrorcondition}, we can write the distribution of the residual term $\mathbf{R}_{m,\omega,r,k}$ from that of $\mathbf{U}_{m,\omega,r}$.
		\begin{align}
			\mathbf{U}_{m,\omega,r} &\sim \mathcal{MN}_{T\times d}\left(\mathbf{0},\mathbf{I}_{T},\mathbf{S}_{m,\omega,r}\right)\\
			\mathbf{R}_{m,\omega,r,k} &= \mathbf{H}_{m,\omega,r,k} \left[\mathbf{H}_{m,\omega,r,-k}'\mathbf{H}_{m,\omega,r,-k}\right]^{-1}\mathbf{H}_{m,\omega,r,-k}'\mathbf{Y}_{\omega,-k} - \mathbf{Y}_{\omega,k}\\
			\mathbf{R}_{m,\omega,r,k} &\sim \mathcal{MN}_{T_k \times d}\left(\mathbf{0}, \mathbf{H}_{m,\omega,r,k} \left[\mathbf{H}_{m,\omega,r,-k}'\mathbf{H}_{m,\omega,r,-k}\right]^{-1}\mathbf{H}_{m,\omega,r,k}' + \mathbf{I}_{T_{k}}, \mathbf{S}_{m,\omega,r}\right)\label{Rdistribution}
		\end{align}
		Define $\bm{\Phi}_{m,\omega,r,k} = \mathbf{H}_{m,\omega,r,k}\left[\mathbf{H}_{m,\omega,r,-k}'\mathbf{H}_{m,\omega,r,-k}\right]^{-1}\mathbf{H}_{m,\omega,r,k}'$, and further define $\mathbf{V}_{m,\omega,r,k}\sim\mathcal{MN}_{T_{k}\times d}(\mathbf{0}, \mathbf{I}_{T_{k}}, \mathbf{I}_d)$ such that we can write $\mathbf{R}_{m,\omega,r,k} = \left(\bm{\Phi}_{m,\omega,r,k} + \mathbf{I}_{T_{k}}\right)^{1/2}\mathbf{V}_{m,\omega,r,k}\mathbf{S}_{m,\omega,r}^{1/2}$.
		\begin{align}
			&\mathbb{E}\left[\text{tr}\left(\mathbf{R}'_{m,\omega,r,k}\mathbf{R}_{m,\omega,r,k}\right)|\mathbf{S}_{m,\omega,r}\right]\nonumber\\
			&\hskip1in= \mathbb{E}\left[\text{tr}\left(\mathbf{S}_{m,\omega,r}^{1/2}\mathbf{V}'_{m,\omega,r,k}\left[\bm{\Phi}_{m,\omega,r,k} + \mathbf{I}_{T_{k}}\right]\mathbf{V}_{m,\omega,r,k}\mathbf{S}_{m,\omega,r}^{1/2}\right)|\mathbf{S}_{m,\omega,r}\right]\label{decomp1}\\
			&\hskip1in= \mathbb{E}\left[\text{tr}\left(\mathbf{S}_{m,\omega,r}\mathbf{V}'_{m,\omega,r,k}\left[\bm{\Phi}_{m,\omega,r,k} + \mathbf{I}_{T_{k}}\right]\mathbf{V}_{m,\omega,r,k}\right) |\mathbf{S}_{m,\omega,r}\right]\\
			&\hskip1in= \text{tr}\left(\mathbb{E}\left[\mathbf{S}_{m,\omega,r}\mathbf{V}'_{m,\omega,r,k}\left[\bm{\Phi}_{m,\omega,r,k} + \mathbf{I}_{T_{k}}\right]\mathbf{V}_{m,\omega,r,k} |\mathbf{S}_{m,\omega,r}\right]\right)\\
			&\hskip1in= \text{tr}\left(\mathbf{S}_{m,\omega,r}\mathbb{E}\left[\mathbf{V}'_{m,\omega,r,k}\left[\bm{\Phi}_{m,\omega,r,k} + \mathbf{I}_{T_{k}}\right]\mathbf{V}_{m,\omega,r,k} |\mathbf{S}_{m,\omega,r}\right]\right)\label{dh1}
		\end{align}
		We can write 
		\begin{align}
			&\mathbb{E}\left[\mathbf{V}'_{m,\omega,r,k}\left[\bm{\Phi}_{m,\omega,r,k} + \mathbf{I}_{T_{k}}\right]\mathbf{V}_{m,\omega,r,k}|\mathbf{S}_{m,\omega,r}\right]\nonumber\\
			&\hskip1in= \mathbb{E}\left[\mathbb{E}\left(\mathbf{V}'_{m,\omega,r,k}\left[\bm{\Phi}_{m,\omega,r,k} + \mathbf{I}_{T_{k}}\right]\mathbf{V}_{m,\omega,r,k}|\bm{\Phi}_{m,\omega,r,k}, \mathbf{S}_{m,\omega,r}\right)\big|\mathbf{S}_{m,\omega,r}\right],
		\end{align}
		and from \citet{gupta18},
		\begin{align}
			&\mathbb{E}\left(\mathbf{V}'_{m,\omega,r,k}\left[\bm{\Phi}_{m,\omega,r,k} + \mathbf{I}_{T_{k}}\right]\mathbf{V}_{m,\omega,r,k}|\bm{\Phi}_{m,\omega,r,k}, \mathbf{S}_{m,\omega,r}\right)]\nonumber\\
			&\hskip1in= \text{tr}\left(\mathbf{I}_{T_{k}}\left[\bm{\Phi}_{m,\omega,r,k} + \mathbf{I}_{T_{k}}\right]\right)\mathbf{I}_d\\
			&\hskip1in= \text{tr}\left(\bm{\Phi}_{m,\omega,r,k} + \mathbf{I}_{T_{k}}\right)\mathbf{I}_d.
		\end{align}
		Thus,
		\begin{align}
			\mathbb{E}\left[\mathbf{V}'_{m,\omega,r,k}\left[\bm{\Phi}_{m,\omega,r,k} + \mathbf{I}_{T_{k}}\right]\mathbf{V}_{m,\omega,r,k}|\mathbf{S}_{m,\omega,r}\right] &= \mathbb{E}\left[\text{tr}\left(\bm{\Phi}_{m,\omega,r,k} + \mathbf{I}_{T_{k}}\right)\mathbf{I}_d|\mathbf{S}_{m,\omega,r}\right]\\
			&= \mathbb{E}\left[\text{tr}\left(\bm{\Phi}_{m,\omega,r,k}\right)|\mathbf{S}_{m,\omega,r}\right]\mathbf{I}_d + T_{k}\mathbf{I}_d.\label{dh2}
		\end{align}
		We modify Equation \ref{dh2} by exchanging $\bm{\Phi}_{m,\omega,r,k}$ for its definition as written above.
		\begin{align}
			\mathbb{E}\left[\text{tr}\left(\bm{\Phi}_{m,\omega,r,k}\right)|\mathbf{S}_{m,\omega,r}\right] &= \mathbb{E}\left[\text{tr}\left(\mathbf{H}_{m,\omega,r,k}\left[\mathbf{H}_{m,\omega,r,-k}'\mathbf{H}_{m,\omega,r,-k}\right]^{-1}\mathbf{H}_{m,\omega,r,k}'\right)|\mathbf{S}_{m,\omega,r}\right]\\
			&= \mathbb{E}\left[\text{tr}\left(\left[\mathbf{H}_{m,\omega,r,-k}'\mathbf{H}_{m,\omega,r,-k}\right]^{-1}\mathbf{H}_{m,\omega,r,k}'\mathbf{H}_{m,\omega,r,k}\right)|\mathbf{S}_{m,\omega,r}\right]\label{dh3}
		\end{align}
		\vskip0in
		We examine the matrices $\mathbf{H}_{m,\omega,r,-k}$ and $\mathbf{H}_{m,\omega,r,k}$ under Conditions \ref{functioncondition}, \ref{rankcondition}, \ref{entriescondition}, and \ref{approximationcondition}. Recall $T_k$ is the number of rows in the test set, and $T_{-k}$ the number in the training set. We obtain a fixed and finite $N$ from Condition \ref{approximationcondition}, and $\text{rank}(\mathbf{H}_{m,\omega,r,-k})= N$ under Equation \ref{rankconditionequation} from Condition \ref{rankcondition}. We label the singular values of the matrices $\sigma_i(\mathbf{H}_{m,\omega,r,-k})>0$ and $\sigma_i(\mathbf{H}_{m,\omega,r,k})\geq 0$ for $i=1,\ldots,N$. For any matrix $\mathbf{P}$, $\left\|\mathbf{P}\right\|_F = \sqrt{\sum_{i=1}^{n_1}\sum_{j=1}^{n_2}|p_{ij}|^2} = \sqrt{\sum_{i=1}^{\text{rank}(\mathbf{P})}\sigma_i(\mathbf{P})^2}$. From Condition \ref{functioncondition}, each entry in $\mathbf{H}_{m,\omega,r,-k}$ and $\mathbf{H}_{m,\omega,r,k}$ lies in a bounded interval, $h_{m,\omega,r,-k,ij},h_{m,\omega,r,k,ij}\in[a,b]$ such that $|h_{m,\omega,r,-k,ij}|,|h_{m,\omega,r,k,ij}|\leq G$, and we use this to establish the upper bounds for the maximal singular values shown in Equations \ref{maxsvtrain} and \ref{maxsvtest}. To write the minimum values, we define the average squared matrix entry as in Condition \ref{entriescondition}, and note the combination of Conditions \ref{functioncondition}, \ref{rankcondition}, and \ref{entriescondition} yields $0<\nu^2\leq 1$, $1\leq \xi^2<\infty$, and $1\leq \varrho^4 < \infty$.
		\begin{align}
			\sqrt{T_{-k}\bar{h}_{m,\omega,r,-k}^2}&\leq \sigma_1\left(\mathbf{H}_{m,\omega,r,-k}\right) \leq \left\|\mathbf{H}_{m,\omega,r,-k}\right\|_F \leq \sqrt{T_{-k}NG^2}\label{maxsvtrain}\\
			\sqrt{T_k\bar{h}^2_{m,\omega,r,k}}&\leq\sigma_1\left(\mathbf{H}_{m,\omega,r,k}\right) \leq \left\|\mathbf{H}_{m,\omega,r,k}\right\|_F \leq \sqrt{T_{k}NG^2}\label{maxsvtest}
		\end{align}
		\vskip0in
		Under Condition \ref{rankcondition}, $\mathbf{H}_{m,\omega,r,-k}$ has a finite condition number $\kappa\left(\mathbf{H}_{m,\omega,r,-k}\right)\leq\kappa_{\max}$, and we can use the bounds in Equations \ref{maxsvtrain} and \ref{maxsvtest} to obtain bounds for the minimum singular values. The rank of $\mathbf{H}_{m,\omega,r,k}$ is not necessarily $N$, as $T_k$ can be less than $N$, but it is at least $\min\left\{T_k,N\right\}$.
		\begin{align}
			\frac{\sqrt{T_{-k}\bar{h}_{m,\omega,r,-k}^2}}{\kappa_{\max}}&\leq \sigma_N\left(\mathbf{H}_{m,\omega,r,-k}\right) \leq  \sqrt{T_{-k}\bar{h}_{m,\omega,r,-k}^2}\label{minsvtrain}\\
			0 &\leq \sigma_N\left(\mathbf{H}_{m,\omega,r,k}\right) \leq \sqrt{T_k\bar{h}^2_{m,\omega,r,k}}\label{minsvtest}
		\end{align}
		\vskip0in
		We return to the $\text{tr}\left(\left[\mathbf{H}_{m,\omega,r,-k}'\mathbf{H}_{m,\omega,r,-k}\right]^{-1}\mathbf{H}_{m,\omega,r,k}'\mathbf{H}_{m,\omega,r,k}\right)$ piece of Equation \ref{dh3}. Applying the above results, we can write $\left[\mathbf{H}_{m,\omega,r,-k}'\mathbf{H}_{m,\omega,r,-k}\right]^{-1}\succ 0$, $\mathbf{H}_{m,\omega,r,k}'\mathbf{H}_{m,\omega,r,k}\succeq 0$, and the following bounds on the extreme eigenvalues.
		\begin{align}
			\frac{1}{T_{-k}\bar{h}_{m,\omega,r,-k}^2} &\leq \lambda_1\left(\left[\mathbf{H}_{m,\omega,r,-k}'\mathbf{H}_{m,\omega,r,-k}\right]^{-1}\right) \leq \frac{\kappa_{\max}^2}{T_{-k}\bar{h}_{m,\omega,r,-k}^2}\label{traininverseeigen1}\\
			\frac{1}{T_{-k}NG^2} &\leq \lambda_N\left(\left[\mathbf{H}_{m,\omega,r,-k}'\mathbf{H}_{m,\omega,r,-k}\right]^{-1}\right) \leq \frac{1}{T_{-k}\bar{h}_{m,\omega,r,-k}^2}\label{traininverseeigen2}\\
			T_k\bar{h}^2_{m,\omega,r,k} &\leq \lambda_1\left(\mathbf{H}_{m,\omega,r,k}'\mathbf{H}_{m,\omega,r,k}\right) \leq T_kNG^2\label{testeigen1}\\
			0 &\leq \lambda_N\left(\mathbf{H}_{m,\omega,r,k}'\mathbf{H}_{m,\omega,r,k}\right) \leq T_k\bar{h}^2_{m,\omega,r,k}\label{testeigen2}
		\end{align}
		\vskip0in
		We can bound the trace of the product with the eigenvalues using Von Neumann's trace inequality.
		\begin{align}
			&\text{tr}\left(\left[\mathbf{H}_{m,\omega,r,-k}'\mathbf{H}_{m,\omega,r,-k}\right]^{-1}\mathbf{H}_{m,\omega,r,k}'\mathbf{H}_{m,\omega,r,k}\right)\nonumber\\
			&\hskip1in\leq \sum_{i=1}^N \lambda_i\left(\left[\mathbf{H}_{m,\omega,r,-k}'\mathbf{H}_{m,\omega,r,-k}\right]^{-1}\right) \lambda_i\left(\mathbf{H}_{m,\omega,r,k}'\mathbf{H}_{m,\omega,r,k}\right)\\
			&\hskip1in\leq N\lambda_1\left(\left[\mathbf{H}_{m,\omega,r,-k}'\mathbf{H}_{m,\omega,r,-k}\right]^{-1}\right) \lambda_1\left(\mathbf{H}_{m,\omega,r,k}'\mathbf{H}_{m,\omega,r,k}\right)\\
			&\hskip1in\leq \frac{T_kN^2G^2\kappa_{\max}^2}{T_{-k}\bar{h}_{m,\omega,r,-k}^2}\label{tracebound}
		\end{align}
		The upper bound in Equation \ref{tracebound} is a positive, finite constant. Returning to Equation \ref{dh3},
		\begin{align}
			\mathbb{E}\left[\text{tr}\left(\bm{\Phi}_{m,\omega,r,k}\right)|\mathbf{S}_{m,\omega,r}\right] &\leq \mathbb{E}\left[\frac{T_kN^2G^2\kappa_{\max}^2}{T_{-k}\bar{h}_{m,\omega,r,-k}^2}\big|\mathbf{S}_{m,\omega,r}\right]\\
			&\leq \frac{T_k}{T_{-k}}\left(NG\kappa_{\max}\right)^2\mathbb{E}\left[\left(\bar{h}_{m,\omega,r,-k}^2\right)^{-1}|\mathbf{S}_{m,\omega,r}\right]\\
			&\leq \frac{T_k}{T_{-k}}\left(N\xi\kappa_{\max}\right)^2.
		\end{align}
		Using this result in Equation \ref{dh2},
		\begin{align}
			\mathbb{E}\left[\mathbf{V}'_{m,\omega,r,k}\left[\bm{\Phi}_{m,\omega,r,k} + \mathbf{I}_{T_{k}}\right]\mathbf{V}_{m,\omega,r,k}|\mathbf{S}_{m,\omega,r}\right] &\leq \frac{T_k}{T_{-k}}\left(N\xi\kappa_{\max}\right)^2\mathbf{I}_d + T_k\mathbf{I}_d,
		\end{align}
		and in Equation \ref{dh1},
		\begin{align}
			\mathbb{E}\left[\text{tr}\left(\mathbf{R}'_{m,\omega,r,k}\mathbf{R}_{m,\omega,r,k}\right)|\mathbf{S}_{m,\omega,r}\right] &\leq \text{tr}\left(\mathbf{S}_{m,\omega,r}\left[\frac{T_k}{T_{-k}}\left(N\xi\kappa_{\max}\right)^2\mathbf{I}_d + T_k\mathbf{I}_d\right]\right)\\
			&\leq \left[\frac{T_k}{T_{-k}}\left(N\xi\kappa_{\max}\right)^2 + T_k\right]\text{tr}\left(\mathbf{S}_{m,\omega,r}\right)\label{help1}.
		\end{align}
		\vskip0in
		We now require a lower bound for the expectation. We examine the trace inequality
		\begin{align}
			&\text{tr}\left(\left[\mathbf{H}_{m,\omega,r,-k}'\mathbf{H}_{m,\omega,r,-k}\right]^{-1}\mathbf{H}_{m,\omega,r,k}'\mathbf{H}_{m,\omega,r,k}\right)\nonumber\\
			&\hskip1in\geq \sum_{i=1}^N \lambda_i\left(\left[\mathbf{H}_{m,\omega,r,-k}'\mathbf{H}_{m,\omega,r,-k}\right]^{-1}\right) \lambda_{N-i+1}\left(\mathbf{H}_{m,\omega,r,k}'\mathbf{H}_{m,\omega,r,k}\right)\\
			&\hskip1in\geq \lambda_N\left(\left[\mathbf{H}_{m,\omega,r,-k}'\mathbf{H}_{m,\omega,r,-k}\right]^{-1}\right) \lambda_1\left(\mathbf{H}_{m,\omega,r,k}'\mathbf{H}_{m,\omega,r,k}\right)\\
			&\hskip1in\geq \frac{T_k\bar{h}^2_{m,\omega,r,k}}{T_{-k}NG^2}.\label{lowertracebound}
		\end{align}
		Again returning to Equation \ref{dh3},
		\begin{align}
			\mathbb{E}\left[\text{tr}\left(\bm{\Phi}_{m,\omega,r,k}\right)|\mathbf{S}_{m,\omega,r}\right] &\geq \mathbb{E}\left[\frac{T_k\bar{h}^2_{m,\omega,r,k}}{T_{-k}NG^2}\big|\mathbf{S}_{m,\omega,r}\right]\\
			&\geq \frac{T_k}{T_{-k}NG^2}\mathbb{E}\left[\bar{h}_{m,\omega,r,k}^2|\mathbf{S}_{m,\omega,r}\right]\\
			&\geq \frac{T_k\nu^2}{T_{-k}N},
		\end{align}
		and using the result in Equation \ref{dh2},
		\begin{align}
			\mathbb{E}\left[\mathbf{V}'_{m,\omega,r,k}\left[\bm{\Phi}_{m,\omega,r,k} + \mathbf{I}_{T_{k}}\right]\mathbf{V}_{m,\omega,r,k}|\mathbf{S}_{m,\omega,r}\right] &\geq \frac{T_k\nu^2}{T_{-k}N}\mathbf{I}_d + T_k\mathbf{I}_d.
		\end{align}
		This yields the lower bound,
		\begin{align}
			\mathbb{E}\left[\text{tr}\left(\mathbf{R}'_{m,\omega,r,k}\mathbf{R}_{m,\omega,r,k}\right)|\mathbf{S}_{m,\omega,r}\right] &\geq \text{tr}\left(\mathbf{S}_{m,\omega,r}\left[\frac{T_k\nu^2}{T_{-k}N}\mathbf{I}_d + T_k\mathbf{I}_d\right]\right)\\
			&\geq \left[\frac{T_k\nu^2}{T_{-k}N} + T_k\right]\text{tr}\left(\mathbf{S}_{m,\omega,r}\right)\label{help2}.
		\end{align}
		\vskip0in
		We return to the expectation equation given in Equation \ref{fullexp1} with the two bounds from Equations \ref{help1} and \ref{help2}.
		\begin{align}
			&\frac{1}{\varphi\mathscr{R}K}\sum_{\omega=1}^{\varphi}\sum_{r=1}^{\mathscr{R}}\sum_{k=1}^K\frac{1}{T_k}\mathbb{E}\left[\mathbb{E}\left(\left[\frac{T_k\nu^2}{T_{-k}N} + T_k\right]\text{tr}\left(\mathbf{S}_{m,\omega,r}\right)|\bm{\Sigma}_{m,\omega}\right)\right]\leq \mathbb{E}\left[\hat{\vartheta}_m\right] \nonumber\\ 
			&\hskip0.5in \leq\frac{1}{\varphi\mathscr{R}K}\sum_{\omega=1}^{\varphi}\sum_{r=1}^{\mathscr{R}}\sum_{k=1}^K\frac{1}{T_k}\mathbb{E}\left[\mathbb{E}\left(\left[\frac{T_k}{T_{-k}}\left(N\xi\kappa_{\max}\right)^2 + T_k\right]\text{tr}\left(\mathbf{S}_{m,\omega,r}\right)|\bm{\Sigma}_{m,\omega}\right)\right] \\
			&\frac{1}{\varphi\mathscr{R}K}\sum_{\omega=1}^{\varphi}\sum_{r=1}^{\mathscr{R}}\sum_{k=1}^K\frac{1}{T_k}\mathbb{E}\left[\left(\frac{T_k\nu^2}{T_{-k}N} + T_k\right)\text{tr}\left(\bm{\Sigma}_{m,\omega}\right)\right]\leq \mathbb{E}\left[\hat{\vartheta}_m\right] \nonumber\\ 
			&\hskip0.5in \leq\frac{1}{\varphi\mathscr{R}K}\sum_{\omega=1}^{\varphi}\sum_{r=1}^{\mathscr{R}}\sum_{k=1}^K\frac{1}{T_k}\mathbb{E}\left[\left(\frac{T_k}{T_{-k}}\left(N\xi\kappa_{\max}\right)^2 + T_k\right)\text{tr}\left(\bm{\Sigma}_{m,\omega}\right)\right]\\
			&\frac{1}{\varphi\mathscr{R}K}\sum_{\omega=1}^{\varphi}\sum_{r=1}^{\mathscr{R}}\sum_{k=1}^K\left[\frac{\nu^2}{T_{-k}N} + 1\right]\vartheta_m \leq \mathbb{E}\left[\hat{\vartheta}_m\right] \nonumber\\ 
			&\hskip0.5in \leq\frac{1}{\varphi\mathscr{R}K}\sum_{\omega=1}^{\varphi}\sum_{r=1}^{\mathscr{R}}\sum_{k=1}^K\left[\frac{1}{T_{-k}}\left(N\xi\kappa_{\max}\right)^2 + 1\right]\vartheta_m
		\end{align}
		\vskip0in
		In the limit as $T\rightarrow\infty$, with a fixed test set size $T_k<\infty$, $T_{-k} = T-T_k \rightarrow\infty$.
		\begin{align}
			&\lim_{T\rightarrow\infty} \frac{1}{\varphi\mathscr{R}K} \sum_{\omega=1}^{\varphi}\sum_{r=1}^{\mathscr{R}}\sum_{k=1}^K\left[\frac{\nu^2}{T_{-k}N} + 1\right]\vartheta_m \nonumber\\
			&\hskip1in= \frac{1}{\varphi\mathscr{R}K} \sum_{\omega=1}^{\varphi}\sum_{r=1}^{\mathscr{R}}\sum_{k=1}^K\lim_{T_{-k}\rightarrow\infty} \left[\frac{\nu^2}{T_{-k}N} + 1\right]\vartheta_m=\vartheta_m\\
			&\lim_{T\rightarrow\infty} \frac{1}{\varphi\mathscr{R}K} \sum_{\omega=1}^{\varphi}\sum_{r=1}^{\mathscr{R}}\sum_{k=1}^K\left[\frac{1}{T_{-k}}\left(N\xi\kappa_{\max}\right)^2 + 1\right]\vartheta_m \nonumber\\
			&\hskip1in= \frac{1}{\varphi\mathscr{R}K} \sum_{\omega=1}^{\varphi}\sum_{r=1}^{\mathscr{R}}\sum_{k=1}^K\lim_{T_{-k}\rightarrow\infty}\left[\frac{1}{T_{-k}}\left(N\xi\kappa_{\max}\right)^2 + 1\right]\vartheta_m=\vartheta_m
		\end{align}
		Thus, $\lim_{T\rightarrow\infty} \mathbb{E}\left[\hat{\vartheta}_m\right] = \vartheta_m$.
	\end{proof}
	
	\begin{lemma}\label{lemma2}
		Under the listed conditions, $\lim_{T\rightarrow\infty} \text{Var}\left(\hat{\vartheta}_{m}\right)= \varphi^{-1}\tau_\omega^2 + \left(\varphi\mathscr{R}\right)^{-1}\tau_r^2$. 
	\end{lemma}
	\par
	For use in the proof of Lemma \ref{lemma2}, we first state and prove Remark \ref{remark1}, and then proceed like in the proof of Lemma \ref{lemma1}.
	\begin{remark}\label{remark1}
		For any square, positive semidefinite matrix $\mathbf{P}\in\mathbb{R}^{n\times n}$, $\text{tr}^2\left(\mathbf{P}\right)\leq (n^2-n+1)\;\text{tr}\left(\mathbf{P}^2\right)$.
	\end{remark}
	\begin{proof}[Proof of Remark \ref{remark1}]
		Let $\lambda_1\left(\mathbf{P}\right),\lambda_2\left(\mathbf{P}\right),\ldots,\lambda_n\left(\mathbf{P}\right)\geq 0$ be the eigenvalues of $\mathbf{P}\in\mathbb{R}^{n\times n}$. For any $n\geq1$ we can write,
		\begin{align}
			\text{tr}^2\left(\mathbf{P}\right) &= \left(\sum_{i=1}^n\lambda_i(\mathbf{P})\right)^2\\
			&= \sum_{i=1}^n\lambda_i\left(\mathbf{P}\right)^2 + 2\sum_{i=1}^n\sum_{j=1}^{i-1}\lambda_i\left(\mathbf{P}\right)\lambda_j\left(\mathbf{P}\right)\\
			&\leq \sum_{i=1}^n\lambda_i\left(\mathbf{P}\right)^2 + n(n-1)\sum_{i=1}^n\lambda_i\left(\mathbf{P}\right)^2\\
			&\leq \left(n^2-n+1\right)\sum_{i=1}^n\lambda_i\left(\mathbf{P}\right)^2
		\end{align}
		with the last line equal to $(n^2-n+1)\;\text{tr}\left(\mathbf{P}^2\right)$. Thus, $\text{tr}^2\left(\mathbf{P}\right)\leq (n^2-n+1)\;\text{tr}\left(\mathbf{P}^2\right)$.
	\end{proof}
	\begin{proof}[Proof of Lemma \ref{lemma2}]
		We isolate the quantity $\text{tr}\left(\mathbf{R}'_{m,\omega,r,k}\mathbf{R}_{m,\omega,r,k}\right)$ like in the proof of Lemma \ref{lemma1} and examine the conditional variance $\text{Var}\left(\text{tr}\left[\mathbf{R}_{m,\omega,r,k}'\mathbf{R}_{m,\omega,r,k}\right]|\mathbf{S}_{m,\omega,r}\right)$. The unconditional variance follows from the law of total variance and from the definition in Equation \ref{parameter}. 
		\begin{align}
			\text{Var}\left(\hat{\vartheta}_m\right) &= \mathbb{E}\left[\mathbb{E}\left(\text{Var}\left[\frac{1}{\varphi\mathscr{R}K}\sum_{\omega = 1}^\varphi \sum_{r=1}^{\mathscr{R}}\sum_{k=1}^{K}\frac{1}{T_k}\text{tr}\left(\mathbf{R}'_{m,\omega,r,k}\mathbf{R}_{m,\omega,r,k}\right)|\mathbf{S}_{m,\omega,r}\right]|\bm{\Sigma}_{m,\omega}\right)\right] \nonumber\\
			&\hskip0.2in + \mathbb{E}\left[\text{Var}\left(\mathbb{E}\left[\frac{1}{\varphi\mathscr{R}K}\sum_{\omega = 1}^\varphi \sum_{r=1}^{\mathscr{R}}\sum_{k=1}^{K}\frac{1}{T_k}\text{tr}\left(\mathbf{R}'_{m,\omega,r,k}\mathbf{R}_{m,\omega,r,k}\right)|\mathbf{S}_{m,\omega,r}\right]|\bm{\Sigma}_{m,\omega}\right)\right]\nonumber\\
			&\hskip0.2in + \text{Var}\left(\mathbb{E}\left[\mathbb{E}\left(\frac{1}{\varphi\mathscr{R}K}\sum_{\omega = 1}^\varphi \sum_{r=1}^{\mathscr{R}}\sum_{k=1}^{K}\frac{1}{T_k}\text{tr}\left[\mathbf{R}'_{m,\omega,r,k}\mathbf{R}_{m,\omega,r,k}\right]|\mathbf{S}_{m,\omega,r}\right)|\bm{\Sigma}_{m,\omega}\right]\right)\label{fullvar1}
		\end{align}
		We first examine $\text{Var}\left(\text{tr}\left[\mathbf{R}_{m,\omega,r,k}'\mathbf{R}_{m,\omega,r,k}\right]|\mathbf{S}_{m,\omega,r}\right)$.
		\begin{align}
			&\text{Var}\left(\text{tr}\left[\mathbf{R}_{m,\omega,r,k}'\mathbf{R}_{m,\omega,r,k}\right]|\mathbf{S}_{m,\omega,r}\right) = \mathbb{E}\left[\text{tr}^2\left(\mathbf{R}_{m,\omega,r,k}'\mathbf{R}_{m,\omega,r,k}\right)|\mathbf{S}_{m,\omega,r}\right] \nonumber\\
			&\hskip3in - \mathbb{E}\left[\text{tr}\left(\mathbf{R}_{m,\omega,r,k}'\mathbf{R}_{m,\omega,r,k}\right)|\mathbf{S}_{m,\omega,r}\right]^2\label{vardecomp}
		\end{align} 
		To establish an bounds on the variance, we seek both lower and upper bounds for the terms in Equation \ref{vardecomp}. The bounds for the second term can be directly obtained from the derivation in the Proof of Lemma \ref{lemma1}.
		\begin{align}
			&\left[\frac{T_k\nu^2}{T_{-k}N} + T_k\right]^2\text{tr}^2\left(\mathbf{S}_{m,\omega,r}\right) \leq \mathbb{E}\left[\text{tr}\left(\mathbf{R}_{m,\omega,r,k}'\mathbf{R}_{m,\omega,r,k}\right)|\mathbf{S}_{m,\omega,r}\right]^2 \nonumber\\
			&\hskip1.5in \leq \left[\frac{T_k}{T_{-k}}\left(N\xi\kappa_{\max}\right)^2 + T_k\right]^2\text{tr}^2\left(\mathbf{S}_{m,\omega,r}\right)\\
			&\left[\frac{T_k^2\nu^4}{T_{-k}^2N^2} + 2\frac{T_k^2\nu^2}{T_{-k}N} + T_k^2\right]\text{tr}^2\left(\mathbf{S}_{m,\omega,r}\right) \leq \mathbb{E}\left[\text{tr}\left(\mathbf{R}_{m,\omega,r,k}'\mathbf{R}_{m,\omega,r,k}\right)|\mathbf{S}_{m,\omega,r}\right]^2\nonumber\\
			&\hskip1.5in \leq \left[\frac{T_k^2}{T_{-k}^2}\left(N\xi\kappa_{\max}\right)^4 + 2\frac{T_k^2}{T_{-k}}\left(N\xi\kappa_{\max}\right)^2 + T_k^2\right]\text{tr}^2\left(\mathbf{S}_{m,\omega,r}\right)\label{bounds2}
		\end{align}
		\vskip0in
		For the first term in Equation \ref{vardecomp}, we follow the strategy of the proof of Lemma \ref{lemma1}. The matrix in the trace of Equation \ref{vdecomp1} is square and positive semidefinite, and we apply the result of Remark \ref{remark1}.
		\begin{align}
			&\mathbb{E}\left[\text{tr}^2\left(\mathbf{R}'_{m,\omega,r,k}\mathbf{R}_{m,\omega,r,k}\right)|\mathbf{S}_{m,\omega,r}\right]\nonumber\\
			&\hskip0.5in= \mathbb{E}\left[\text{tr}^2\left(\mathbf{S}_{m,\omega,r}^{1/2}\mathbf{V}'_{m,\omega,r,k}\left[\bm{\Phi}_{m,\omega,r,k} + \mathbf{I}_{T_{k}}\right]\mathbf{V}_{m,\omega,r,k}\mathbf{S}_{m,\omega,r}^{1/2}\right)|\mathbf{S}_{m,\omega,r}\right]\label{vdecomp1}\\
			&\hskip0.5in \leq (d^2-d+1)\;\mathbb{E}\bigg[\text{tr}\bigg(\mathbf{S}_{m,\omega,r}\mathbf{V}'_{m,\omega,r,k}\left[\bm{\Phi}_{m,\omega,r,k} + \mathbf{I}_{T_{k}}\right]\mathbf{V}_{m,\omega,r,k}\mathbf{S}_{m,\omega,r}\nonumber\\
			&\hskip2.5in\mathbf{V}'_{m,\omega,r,k}\left[\bm{\Phi}_{m,\omega,r,k} + \mathbf{I}_{T_{k}}\right]\mathbf{V}_{m,\omega,r,k}\bigg)|\mathbf{S}_{m,\omega,r}\bigg]\\
			&\hskip0.5in \leq (d^2-d+1)\;\text{tr}\bigg(\mathbf{S}_{m,\omega,r}\mathbb{E}\bigg[\mathbf{V}'_{m,\omega,r,k}\left[\bm{\Phi}_{m,\omega,r,k} + \mathbf{I}_{T_{k}}\right]\mathbf{V}_{m,\omega,r,k}\mathbf{S}_{m,\omega,r}\nonumber\\
			&\hskip2.5in\mathbf{V}'_{m,\omega,r,k}\left[\bm{\Phi}_{m,\omega,r,k} + \mathbf{I}_{T_{k}}\right]\mathbf{V}_{m,\omega,r,k}|\mathbf{S}_{m,\omega,r}\bigg]\bigg)\label{help3}
		\end{align}
		We can write the expectation as
		\begin{align}
			&\mathbb{E}\bigg[\mathbf{V}'_{m,\omega,r,k}\left[\bm{\Phi}_{m,\omega,r,k} + \mathbf{I}_{T_{k}}\right]\mathbf{V}_{m,\omega,r,k}\mathbf{S}_{m,\omega,r}\mathbf{V}'_{m,\omega,r,k}\left[\bm{\Phi}_{m,\omega,r,k} + \mathbf{I}_{T_{k}}\right]\mathbf{V}_{m,\omega,r,k}|\mathbf{S}_{m,\omega,r}\bigg]\nonumber\\
			&\hskip0.5in= \mathbb{E}\bigg[\mathbb{E}\bigg(\mathbf{V}'_{m,\omega,r,k}\left[\bm{\Phi}_{m,\omega,r,k} + \mathbf{I}_{T_{k}}\right]\mathbf{V}_{m,\omega,r,k}\mathbf{S}_{m,\omega,r}\mathbf{V}'_{m,\omega,r,k}\nonumber\\
			&\hskip1.5in\left[\bm{\Phi}_{m,\omega,r,k} + \mathbf{I}_{T_{k}}\right]\mathbf{V}_{m,\omega,r,k}|\bm{\Phi}_{m,\omega,r,k},\mathbf{S}_{m,\omega,r}\bigg)|\mathbf{S}_{m,\omega,r}\bigg]
		\end{align}
		and obtain the inner piece from Theorem 2.3.8 (v) of \citet{gupta18}.
		\begin{align}
			&\mathbb{E}\bigg(\mathbf{V}'_{m,\omega,r,k}\left[\bm{\Phi}_{m,\omega,r,k} + \mathbf{I}_{T_{k}}\right]\mathbf{V}_{m,\omega,r,k}\mathbf{S}_{m,\omega,r}\mathbf{V}'_{m,\omega,r,k}\left[\bm{\Phi}_{m,\omega,r,k} + \mathbf{I}_{T_{k}}\right]\mathbf{V}_{m,\omega,r,k}|\bm{\Phi}_{m,\omega,r,k},\mathbf{S}_{m,\omega,r}\bigg) \nonumber\\
			&\hskip0.5in= \text{tr}\left(\mathbf{I}_{T_k}\left[\bm{\Phi}_{m,\omega,r,k} + \mathbf{I}_{T_{k}}\right]\mathbf{I}_{T_k}\left[\bm{\Phi}_{m,\omega,r,k} + \mathbf{I}_{T_{k}}\right]\right)\text{tr}\left(\mathbf{S}_{m,\omega,r}\mathbf{I}_d\right)\mathbf{I}_d\nonumber\\
			&\hskip1in + \text{tr}\left(\left[\bm{\Phi}_{m,\omega,r,k} + \mathbf{I}_{T_{k}}\right]\mathbf{I}_{T_k}\right)\text{tr}\left(\left[\bm{\Phi}_{m,\omega,r,k} + \mathbf{I}_{T_{k}}\right]\mathbf{I}_{T_k}\right)\mathbf{I}_d\mathbf{S}_{m,\omega,r}\mathbf{I}_d\nonumber\\
			&\hskip1in + \text{tr}\left(\left[\bm{\Phi}_{m,\omega,r,k} + \mathbf{I}_{T_{k}}\right]\mathbf{I}_{T_k}\left[\bm{\Phi}_{m,\omega,r,k} + \mathbf{I}_{T_{k}}\right]\mathbf{I}_{T_k}\right)\mathbf{I}_d\mathbf{S}_{m,\omega,r}\mathbf{I}_d\\
			&\hskip0.5in= \text{tr}\left(\left[\bm{\Phi}_{m,\omega,r,k} + \mathbf{I}_{T_{k}}\right]^2\right)\text{tr}\left(\mathbf{S}_{m,\omega,r}\right)\mathbf{I}_d\nonumber\\
			&\hskip1in + \text{tr}^2\left(\left[\bm{\Phi}_{m,\omega,r,k} + \mathbf{I}_{T_{k}}\right]\right)\mathbf{S}_{m,\omega,r}\nonumber\\
			&\hskip1in + \text{tr}\left(\left[\bm{\Phi}_{m,\omega,r,k} + \mathbf{I}_{T_{k}}\right]^2\right)\mathbf{S}_{m,\omega,r}\\
			&\hskip0.5in= \text{tr}\left(\bm{\Phi}_{m,\omega,r,k}^2\right)\text{tr}\left(\mathbf{S}_{m,\omega,r}\right)\mathbf{I}_d + \text{tr}\left(\bm{\Phi}_{m,\omega,r,k}^2\right)\mathbf{S}_{m,\omega,r} +\text{tr}^2\left(\bm{\Phi}_{m,\omega,r,k}\right)\mathbf{S}_{m,\omega,r} \nonumber\\
			&\hskip1in + 2\text{tr}\left(\bm{\Phi}_{m,\omega,r,k}\right)\text{tr}\left(\mathbf{S}_{m,\omega,r}\right)\mathbf{I}_d + 2\text{tr}\left(\bm{\Phi}_{m,\omega,r,k}\right)\mathbf{S}_{m,\omega,r} + 2T_k\text{tr}\left(\bm{\Phi}_{m,\omega,r,k}\right)\mathbf{S}_{m,\omega,r}\nonumber\\
			&\hskip1in + T_k\text{tr}\left(\mathbf{S}_{m,\omega,r}\right)\mathbf{I}_d + T_k\mathbf{S}_{m,\omega,r} + T_k^2\mathbf{S}_{m,\omega,r}\label{quarticexpectation}
		\end{align}
		Using Von Neumann's trace inequality and the upper bound on the largest eigenvalue from Equation \ref{tracebound},
		\begin{align}
			\text{tr}\left(\bm{\Phi}^2_{m,\omega,r,k}\right)&\leq \sum_{i=1}^N \lambda_i\left(\bm{\Phi}_{m,\omega,r,k}\right)^2\\
			&\leq N\lambda_1\left(\bm{\Phi}_{m,\omega,r,k}\right)^2\\
			&\leq \frac{T_k^2N^3G^4\kappa_{\max}^4}{T_{-k}^2\left(\bar{h}^2_{m,\omega,r,-k}\right)^2},
		\end{align}
		and we input the result to Equation \ref{quarticexpectation}.
		\begin{align}
			&\mathbb{E}\bigg(\mathbf{V}'_{m,\omega,r,k}\left[\bm{\Phi}_{m,\omega,r,k} + \mathbf{I}_{T_{k}}\right]\mathbf{V}_{m,\omega,r,k}\mathbf{S}_{m,\omega,r}\mathbf{V}'_{m,\omega,r,k}\left[\bm{\Phi}_{m,\omega,r,k} + \mathbf{I}_{T_{k}}\right]\mathbf{V}_{m,\omega,r,k}|\bm{\Phi}_{m,\omega,r,k},\mathbf{S}_{m,\omega,r}\bigg) \nonumber\\
			&\hskip0.5in\leq \left(\frac{T_k^2N^3G^4\kappa_{\max}^4}{T_{-k}^2\left(\bar{h}^2_{m,\omega,r,-k}\right)^2}\right)\text{tr}\left(\mathbf{S}_{m,\omega,r}\right)\mathbf{I}_d + \left(\frac{T_k^2N^3G^4\kappa_{\max}^4}{T_{-k}^2\left(\bar{h}^2_{m,\omega,r,-k}\right)^2}\right)\mathbf{S}_{m,\omega,r} \nonumber\\
			&\hskip1in+\left(\frac{T_k^2N^4G^4\kappa_{\max}^4}{T_{-k}^2\left(\bar{h}_{m,\omega,r,-k}^2\right)^2}\right)\mathbf{S}_{m,\omega,r} 
			+ 2\left(\frac{T_kN^2G^2\kappa_{\max}^2}{T_{-k}\bar{h}_{m,\omega,r,-k}^2}\right)\text{tr}\left(\mathbf{S}_{m,\omega,r}\right)\mathbf{I}_d \nonumber\\
			&\hskip1in + 2\left(\frac{T_kN^2G^2\kappa_{\max}^2}{T_{-k}\bar{h}_{m,\omega,r,-k}^2}\right)\mathbf{S}_{m,\omega,r} + 2\left(\frac{T_k^2N^2G^2\kappa_{\max}^2}{T_{-k}\bar{h}_{m,\omega,r,-k}^2}\right)\mathbf{S}_{m,\omega,r}\nonumber\\
			&\hskip1in
			+ T_k\text{tr}\left(\mathbf{S}_{m,\omega,r}\right)\mathbf{I}_d + T_k\mathbf{S}_{m,\omega,r} + T_k^2\mathbf{S}_{m,\omega,r}
		\end{align}
		With the expected values from Condition \ref{entriescondition}, we can simplify the form to that shown in Equation \ref{exp1}.
		\begin{align}
			&\mathbb{E}\bigg[\mathbf{V}'_{m,\omega,r,k}\left[\bm{\Phi}_{m,\omega,r,k} + \mathbf{I}_{T_{k}}\right]\mathbf{V}_{m,\omega,r,k}\mathbf{S}_{m,\omega,r}\mathbf{V}'_{m,\omega,r,k}\left[\bm{\Phi}_{m,\omega,r,k} + \mathbf{I}_{T_{k}}\right]\mathbf{V}_{m,\omega,r,k}|\mathbf{S}_{m,\omega,r}\bigg] \nonumber\\
			&\hskip0.5in\leq \left(\frac{T_k^2N^3\varrho^4\kappa_{\max}^4}{T_{-k}^2}\right)\text{tr}\left(\mathbf{S}_{m,\omega,r}\right)\mathbf{I}_d + \left(\frac{T_k^2N^3\varrho^4\kappa_{\max}^4}{T_{-k}^2}\right)\mathbf{S}_{m,\omega,r} \nonumber\\
			&\hskip1in+\left(\frac{T_k^2N^4\varrho^4\kappa_{\max}^4}{T_{-k}^2}\right)\mathbf{S}_{m,\omega,r} 
			+ 2\left(\frac{T_kN^2\xi^2\kappa_{\max}^2}{T_{-k}}\right)\text{tr}\left(\mathbf{S}_{m,\omega,r}\right)\mathbf{I}_d \nonumber\\
			&\hskip1in + 2\left(\frac{T_kN^2\xi^2\kappa_{\max}^2}{T_{-k}}\right)\mathbf{S}_{m,\omega,r} + 2\left(\frac{T_k^2N^2\xi^2\kappa_{\max}^2}{T_{-k}}\right)\mathbf{S}_{m,\omega,r}\nonumber\\
			&\hskip1in
			+ T_k\text{tr}\left(\mathbf{S}_{m,\omega,r}\right)\mathbf{I}_d + T_k\mathbf{S}_{m,\omega,r} + T_k^2\mathbf{S}_{m,\omega,r}\label{exp1}
		\end{align}
		\vskip0in
		We return to Equation \ref{help3}, and note that for any square, positive semidefinite matrix $\mathbf{P}\in\mathbb{R}^{n\times n}$, $\text{tr}\left(\mathbf{P}^2\right)\leq \text{tr}^2\left(\mathbf{P}\right)$.
		\begin{align}
			&\mathbb{E}\left[\text{tr}^2\left(\mathbf{R}'_{m,\omega,r,k}\mathbf{R}_{m,\omega,r,k}\right)|\mathbf{S}_{m,\omega,r}\right]\nonumber\\
			&\hskip0.5in\leq (d^2 - d + 1)\; \text{tr}^2\left(\mathbf{S}_{m,\omega,r}\right)\bigg[3\left(\frac{T_k^2N^4\varrho^4\kappa_{\max}^4}{T_{-k}^2}\right) + 4\left(\frac{T_kN^2\xi^2\kappa_{\max}^2}{T_{-k}}\right)\nonumber\\
			&\hskip1.5in
			+ 2\left(\frac{T_k^2N^2\xi^2\kappa_{\max}^2}{T_{-k}}\right)+ 2T_k + T_k^2\bigg]
		\end{align}
		We can write the following inequality for $\text{Var}\left(T_k^{-1}\text{tr}\left[\mathbf{R}_{m,\omega,r,k}'\mathbf{R}_{m,\omega,r,k}\right]|\mathbf{S}_{m,\omega,r}\right)$.
		\begin{align}
			&\text{Var}\left(T_k^{-1}\text{tr}\left[\mathbf{R}_{m,\omega,r,k}'\mathbf{R}_{m,\omega,r,k}\right]|\mathbf{S}_{m,\omega,r}\right) \nonumber\\
			&\hskip0.5in \leq \text{tr}^2\left(\mathbf{S}_{m,\omega,r}\right)\bigg(\frac{1}{T_{-k}^2}\left[\frac{3(d^2 - d + 1)N^6\varrho^4\kappa_{\max}^4 - \nu^4}{N^2}\right] \nonumber\\
			&\hskip1in + \frac{1}{T_kT_{-k}}\left[4(d^2-d+1)N^2\xi^2\kappa_{\max}^2\right]\nonumber\\
			&\hskip1in + \frac{1}{T_{-k}}\left[\frac{2(d^2 - d + 1)N^3\xi^2\kappa_{\max}^2 - 2\nu^2}{N}\right] \nonumber\\
			&\hskip1in+ (d^2-d) + \frac{2}{T_K}(d^2 - d + 1)\bigg)\label{dform}
		\end{align}
		\vskip0in
		Taking the limit as $T\rightarrow\infty$, we obtain an upper bound for the variance of the estimate given all individual covariance matrices for one test set with each featurization and potential realization of the data. As the number of time points gets large with a fixed $T_k<\infty$, the size of each training set $T_{-k}$ tends to infinity. 
		\begin{align}
			\lim_{T\rightarrow\infty} \text{Var}\left(T_k^{-1}\text{tr}\left[\mathbf{R}_{m,\omega,r,k}'\mathbf{R}_{m,\omega,r,k}\right]|\mathbf{S}_{m,\omega,r}\right) &\leq  \left[(d^2-d) + \frac{2}{T_k}(d^2 - d + 1)\right]\text{tr}^2\left(\mathbf{S}_{m,\omega,r}\right)\label{singlelimit}
		\end{align}
		\vskip0in
		Returning to the total variance in Equation \ref{fullvar1}, with the limit in Equation \ref{singlelimit}, we can apply the Dominated Convergence Theorem to establish the limiting value for $\text{Var}\left(\hat{\vartheta}_m\right)$. We also use the results in Equations \ref{help1} and \ref{help2} to establish the limiting expectation $\lim_{T\rightarrow\infty}\mathbb{E}\left[T_k^{-1}\text{tr}\left(\mathbf{R}_{m,\omega,r,k}'\mathbf{R}_{m,\omega,r,k}\right)|\mathbf{S}_{m,\omega,r}\right]=\text{tr}\left(\mathbf{S}_{m,\omega,r}\right)$.
		\vskip0in
		The first term in Equation \ref{fullvar1} has an inner term of the variance of the variation parameter given the featurization specific matrix $\mathbf{S}_{m,\omega,r}$ and the realization specific matrix $\bm{\Sigma}_{m,\omega}$. The only remaining source of variation arises from labelling the training and test sets $k=1,\ldots,K$. With random assignment, and fully explained temporal dependence as in Condition \ref{modelerrorcondition}, we treat these splits as uncorrelated draws, and the variance can move inside the summation.
		\begin{align}
			&\mathbb{E}\left[\mathbb{E}\left(\text{Var}\left[\frac{1}{\varphi\mathscr{R}K}\sum_{\omega = 1}^\varphi \sum_{r=1}^{\mathscr{R}}\sum_{k=1}^{K}\frac{1}{T_k}\text{tr}\left(\mathbf{R}'_{m,\omega,r,k}\mathbf{R}_{m,\omega,r,k}\right)|\mathbf{S}_{m,\omega,r}\right]|\bm{\Sigma}_{m,\omega}\right)\right]\nonumber\\
			&\hskip0.5in= \mathbb{E}\left[\mathbb{E}\left(\frac{1}{\varphi^2\mathscr{R}^2K^2}\sum_{\omega = 1}^\varphi \sum_{r=1}^{\mathscr{R}}\sum_{k=1}^{K}\text{Var}\left[\frac{1}{T_k}\text{tr}\left(\mathbf{R}'_{m,\omega,r,k}\mathbf{R}_{m,\omega,r,k}\right)|\mathbf{S}_{m,\omega,r}\right]|\bm{\Sigma}_{m,\omega}\right)\right]
		\end{align}
		The inner variance converges pointwise and is dominated by the integrable quantity 
		\begin{align}
			\text{tr}^2\left(\mathbf{S}_{m,\omega,r}\right)\bigg[3(d^2 - d + 1)N^6\varrho^4\kappa_{\max}^4 
			+ 6(d^2-d+1)N^3\xi^2\kappa_{\max}^2 + 3d^2 + 1\bigg],
		\end{align}
		as a simplification of the form shown in Equation \ref{dform}. The expectation of the inner term converges pointwise, where 
		\begin{align}
			&\lim_{T\rightarrow\infty} \mathbb{E}\left[\text{Var}\left(T_k^{-1}\text{tr}\left[\mathbf{R}_{m,\omega,r,k}'\mathbf{R}_{m,\omega,r,k}\right]|\mathbf{S}_{m,\omega,r}\right)|\bm{\Sigma}_{m,\omega}\right] \nonumber\\
			&\hskip1in \leq  \left[(d^2-d) + \frac{2}{T_k}(d^2 - d + 1)\right]\text{tr}^2\left(\bm{\Sigma}_{m,\omega}\right),
		\end{align}
		and is dominated by the similar integrable quantity
		\begin{align}
			\text{tr}^2\left(\bm{\Sigma}_{m,\omega}\right)\bigg[3(d^2 - d + 1)N^6\varrho^4\kappa_{\max}^4 
			+ 6(d^2-d+1)N^3\xi^2\kappa_{\max}^2 + 3d^2 + 1\bigg].
		\end{align}
		We again assume a fixed $T_k<\infty$, and $T\rightarrow\infty$ implies $T_{-k}\rightarrow\infty$. The limit also implies $K\rightarrow\infty$ for a fixed $T_k$. When applying the Dominated Convergence Theorem twice,
		\begin{align}
			&\lim_{T\rightarrow\infty} \mathbb{E}\left[\mathbb{E}\left(\frac{1}{\varphi^2\mathscr{R}^2K^2}\sum_{\omega = 1}^\varphi \sum_{r=1}^{\mathscr{R}}\sum_{k=1}^{K}\text{Var}\left[\frac{1}{T_k}\text{tr}\left(\mathbf{R}'_{m,\omega,r,k}\mathbf{R}_{m,\omega,r,k}\right)|\mathbf{S}_{m,\omega,r}\right]|\bm{\Sigma}_{m,\omega}\right)\right]\nonumber\\
			&= \mathbb{E}\left[\lim_{T\rightarrow\infty} \mathbb{E}\left(\frac{1}{\varphi^2\mathscr{R}^2K^2}\sum_{\omega = 1}^\varphi \sum_{r=1}^{\mathscr{R}}\sum_{k=1}^{K}\text{Var}\left[\frac{1}{T_k}\text{tr}\left(\mathbf{R}'_{m,\omega,r,k}\mathbf{R}_{m,\omega,r,k}\right)|\mathbf{S}_{m,\omega,r}\right]|\bm{\Sigma}_{m,\omega}\right)\right]\\
			&= \mathbb{E}\left[\mathbb{E}\left(\lim_{T_{-k},K\rightarrow\infty}\frac{1}{\varphi^2\mathscr{R}^2K^2}\sum_{\omega = 1}^\varphi \sum_{r=1}^{\mathscr{R}}\sum_{k=1}^{K}\text{Var}\left[\frac{1}{T_k}\text{tr}\left(\mathbf{R}'_{m,\omega,r,k}\mathbf{R}_{m,\omega,r,k}\right)|\mathbf{S}_{m,\omega,r}\right]|\bm{\Sigma}_{m,\omega}\right)\right]\\
			&\leq \mathbb{E}\left[\mathbb{E}\left(\lim_{K\rightarrow\infty}\frac{1}{\varphi^2\mathscr{R}^2K^2}\sum_{\omega = 1}^\varphi \sum_{r=1}^{\mathscr{R}}\sum_{k=1}^{K}\left[(d^2-d) + \frac{2}{T_k}(d^2 - d + 1)\right]\text{tr}^2\left(\mathbf{S}_{m,\omega,r}\right)|\bm{\Sigma}_{m,\omega}\right)\right].
		\end{align}
		The inner piece tends to zero as $K$ gets large, and 
		\begin{align}
			\lim_{T\rightarrow\infty} \mathbb{E}\left[\mathbb{E}\left(\frac{1}{\varphi^2\mathscr{R}^2K^2}\sum_{\omega = 1}^\varphi \sum_{r=1}^{\mathscr{R}}\sum_{k=1}^{K}\text{Var}\left[\frac{1}{T_k}\text{tr}\left(\mathbf{R}'_{m,\omega,r,k}\mathbf{R}_{m,\omega,r,k}\right)|\mathbf{S}_{m,\omega,r}\right]|\bm{\Sigma}_{m,\omega}\right)\right]=0.
		\end{align}
		\vskip0in
		For the second term in Equation \ref{fullvar1}, the inner expectation converges pointwise to $\text{tr}\left(\mathbf{S}_{m,\omega,r}\right)$ and is dominated by the integrable term $\left[\left(N\xi\kappa_{\max}\right)^2 + 1\right]\text{tr}\left(\mathbf{S}_{m,\omega,r}\right)$. The middle expectation converges pointwise to $\text{tr}\left(\bm{\Sigma}_{m,\omega}\right)$ and is dominated by the similar term $\left[\left(N\xi\kappa_{\max}\right)^2 + 1\right]\text{tr}\left(\bm{\Sigma}_{m,\omega}\right)$. We can take the limit inside the function $h(x)=x^2$ as it is continuous on the full domain $x\in\mathbb{R}$ (\textit{i.e.}, inside the variance term).
		\begin{align}
			&\lim_{T\rightarrow\infty}\mathbb{E}\left[\text{Var}\left(\mathbb{E}\left[\frac{1}{\varphi\mathscr{R}K}\sum_{\omega = 1}^\varphi \sum_{r=1}^{\mathscr{R}}\sum_{k=1}^{K}\frac{1}{T_k}\text{tr}\left(\mathbf{R}'_{m,\omega,r,k}\mathbf{R}_{m,\omega,r,k}\right)|\mathbf{S}_{m,\omega,r}\right]|\bm{\Sigma}_{m,\omega}\right)\right]\nonumber\\
			&=\mathbb{E}\left[\lim_{T\rightarrow\infty}\text{Var}\left(\mathbb{E}\left[\frac{1}{\varphi\mathscr{R}K}\sum_{\omega = 1}^\varphi \sum_{r=1}^{\mathscr{R}}\sum_{k=1}^{K}\frac{1}{T_k}\text{tr}\left(\mathbf{R}'_{m,\omega,r,k}\mathbf{R}_{m,\omega,r,k}\right)|\mathbf{S}_{m,\omega,r}\right]|\bm{\Sigma}_{m,\omega}\right)\right]\\
			&=\mathbb{E}\left[\text{Var}\left(\lim_{T\rightarrow\infty}\frac{1}{\varphi\mathscr{R}K}\sum_{\omega = 1}^\varphi \sum_{r=1}^{\mathscr{R}}\sum_{k=1}^{K}\mathbb{E}\left[\frac{1}{T_k}\text{tr}\left(\mathbf{R}'_{m,\omega,r,k}\mathbf{R}_{m,\omega,r,k}\right)|\mathbf{S}_{m,\omega,r}\right]|\bm{\Sigma}_{m,\omega}\right)\right]\\
			&=\mathbb{E}\left[\text{Var}\left(\frac{1}{\varphi\mathscr{R}}\sum_{\omega = 1}^\varphi \sum_{r=1}^{\mathscr{R}}\text{tr}\left(\mathbf{S}_{m,\omega,r}\right)|\bm{\Sigma}_{m,\omega}\right)\right]
		\end{align}
		Given the data realized covariance matrix $\bm{\Sigma}_{m,\omega}$, each featurization $\mathbf{S}_{m,\omega,r}$ is an uncorrelated draw, and we can distribute the variance term like above. From Condition \ref{distributioncondition}, we extract the limiting value.
		\begin{align}
			&\lim_{T\rightarrow\infty}\mathbb{E}\left[\text{Var}\left(\mathbb{E}\left[\frac{1}{\varphi\mathscr{R}K}\sum_{\omega = 1}^\varphi \sum_{r=1}^{\mathscr{R}}\sum_{k=1}^{K}\frac{1}{T_k}\text{tr}\left(\mathbf{R}'_{m,\omega,r,k}\mathbf{R}_{m,\omega,r,k}\right)|\mathbf{S}_{m,\omega,r}\right]|\bm{\Sigma}_{m,\omega}\right)\right]\nonumber\\
			&=\mathbb{E}\left[\frac{1}{\varphi^2\mathscr{R}^2}\sum_{\omega = 1}^\varphi \sum_{r=1}^{\mathscr{R}}\text{Var}\left(\text{tr}\left(\mathbf{S}_{m,\omega,r}\right)|\bm{\Sigma}_{m,\omega}\right)\right]\\
			&=\mathbb{E}\left[\frac{1}{\varphi^2\mathscr{R}^2}\sum_{\omega = 1}^\varphi \sum_{r=1}^{\mathscr{R}}\tau_r^2\right] = \frac{\tau_r^2}{\varphi\mathscr{R}}
		\end{align}
		\vskip0in
		For the third term in Equation \ref{fullvar1}, we note the same conditions as before, and proceed like above.
		\begin{align}
			&\lim_{T\rightarrow\infty} \text{Var}\left(\mathbb{E}\left[\mathbb{E}\left(\frac{1}{\varphi\mathscr{R}K}\sum_{\omega = 1}^\varphi \sum_{r=1}^{\mathscr{R}}\sum_{k=1}^{K}\frac{1}{T_k}\text{tr}\left[\mathbf{R}'_{m,\omega,r,k}\mathbf{R}_{m,\omega,r,k}\right]|\mathbf{S}_{m,\omega,r}\right)|\bm{\Sigma}_{m,\omega}\right]\right)\nonumber\\
			&= \text{Var}\left(\lim_{T\rightarrow\infty}\mathbb{E}\left[\mathbb{E}\left(\frac{1}{\varphi\mathscr{R}K}\sum_{\omega = 1}^\varphi \sum_{r=1}^{\mathscr{R}}\sum_{k=1}^{K}\frac{1}{T_k}\text{tr}\left[\mathbf{R}'_{m,\omega,r,k}\mathbf{R}_{m,\omega,r,k}\right]|\mathbf{S}_{m,\omega,r}\right)|\bm{\Sigma}_{m,\omega}\right]\right) \\
			&= \text{Var}\left(\mathbb{E}\left[\lim_{T\rightarrow\infty}\frac{1}{\varphi\mathscr{R}K}\sum_{\omega = 1}^\varphi \sum_{r=1}^{\mathscr{R}}\sum_{k=1}^{K}\mathbb{E}\left(\frac{1}{T_k}\text{tr}\left[\mathbf{R}'_{m,\omega,r,k}\mathbf{R}_{m,\omega,r,k}\right]|\mathbf{S}_{m,\omega,r}\right)|\bm{\Sigma}_{m,\omega}\right]\right) \\
			&= \text{Var}\left(\mathbb{E}\left[\frac{1}{\varphi\mathscr{R}}\sum_{\omega = 1}^\varphi \sum_{r=1}^{\mathscr{R}}\text{tr}\left(\mathbf{S}_{m,\omega,r}\right)|\bm{\Sigma}_{m,\omega}\right]\right) \\
			&= \text{Var}\left(\frac{1}{\varphi}\sum_{\omega = 1}^\varphi \text{tr}\left(\bm{\Sigma}_{m,\omega}\right)\right)
		\end{align}
		Each realization of the data is independent, and we can take the variance inside the summation.
		\begin{align}
			&\lim_{T\rightarrow\infty} \text{Var}\left(\mathbb{E}\left[\mathbb{E}\left(\frac{1}{\varphi\mathscr{R}K}\sum_{\omega = 1}^\varphi \sum_{r=1}^{\mathscr{R}}\sum_{k=1}^{K}\frac{1}{T_k}\text{tr}\left[\mathbf{R}'_{m,\omega,r,k}\mathbf{R}_{m,\omega,r,k}\right]|\mathbf{S}_{m,\omega,r}\right)|\bm{\Sigma}_{m,\omega}\right]\right)\nonumber\\
			&= \frac{1}{\varphi^2}\sum_{\omega = 1}^\varphi \text{Var}\left(\text{tr}\left[\bm{\Sigma}_{m,\omega}\right]\right) = \frac{\tau_\omega^2}{\varphi}
		\end{align}
		\vskip0in
		Thus, we can write the limiting variance as,
		\begin{align}
			\lim_{T\rightarrow\infty}  \text{Var}\left(\hat{\vartheta}_m\right) &= \frac{\tau_\omega^2}{\varphi} + \frac{\tau_r^2}{\varphi\mathscr{R}}.
		\end{align}
	\end{proof}
	\par
	The proof of Theorem \ref{theorem1} follows directly from the results of Lemmas \ref{lemma1} and \ref{lemma2}.
	\begin{proof}[Proof of Theorem \ref{theorem1}]
		From Lemma \ref{lemma1}, we establish that $\lim_{T\rightarrow\infty} \mathbb{E}\left[\hat{\vartheta}_m\right] = \vartheta_m$. Turning to the result of Lemma \ref{lemma2},
		\begin{align}
			\lim_{T\rightarrow\infty}  \text{Var}\left(\hat{\vartheta}_m\right) &= \frac{\tau_\omega^2}{\varphi} + \frac{\tau_r^2}{\varphi\mathscr{R}}\\
			\lim_{\mathscr{R}\rightarrow\infty} \lim_{T\rightarrow\infty}  \text{Var}\left(\hat{\vartheta}_m\right) &= \lim_{\mathscr{R}\rightarrow\infty}\frac{\tau_\omega^2}{\varphi} + \frac{\tau_r^2}{\varphi\mathscr{R}} = \frac{\tau_\omega^2}{\varphi}.
		\end{align}
		When we observe the data without error ($\tau_\omega^2 = 0$), the limiting variance is zero, and the estimate is consistent.
		\vskip0in
		In the alternative scenario, we no longer require error free data observation. 
		\begin{align}
			\lim_{\varphi\rightarrow\infty} \lim_{T\rightarrow\infty}  \text{Var}\left(\hat{\vartheta}_m\right) &= \lim_{\varphi\rightarrow\infty}\frac{\tau_\omega^2}{\varphi} + \frac{\tau_r^2}{\varphi\mathscr{R}} = 0.
		\end{align}
	\end{proof}
	\par
	The proof of Theorem \ref{theorem2} is a direct result of Condition \ref{distributioncondition}.
	\begin{proof}[Proof of Theorem \ref{theorem2}]
		From Condition \ref{distributioncondition}, we see that $\text{tr}\left(\bm{\Sigma}_{m,\omega}\right)\overset{\textit{i.i.d.}}{\sim}\left(\vartheta_m, \tau_\omega^2\right)$. For an individual realization, we have shown in the proofs of Lemma \ref{lemma1} and \ref{lemma2} that 
		\begin{align}
			\frac{1}{\mathscr{R}K}\sum_{r=1}^{\mathscr{R}}\sum_{k=1}^KT_k^{-1}\text{tr}\left(\mathbf{R}_{m,\omega,r,k}'\mathbf{R}_{m,\omega,r,k}\right) | \bm{\Sigma}_{m,\omega} \xrightarrow{P} \text{tr}\left(\bm{\Sigma}_{m,\omega}\right)
		\end{align}
		as the number of observations $T\rightarrow\infty$. The variation parameter is the average of the individual realization-specific covariance matrices, and the Central Limit Theorem arises from the application of Slutsky's Theorem.
	\end{proof}
	\par
	The proof of Theorem \ref{theorem3} begins with the empirical distribution function, and employs the result of the Glivenko-Cantelli to show convergence in distribution to a $\text{Uniform}(0,1)$ random variable.
	\begin{proof}[Proof of Theorem \ref{theorem3}]
		Under the null hypothesis, $\vartheta_1 = \cdots = \vartheta_{T!} = \vartheta$. 
		We define the quantile estimate $\hat{Q}_{M}$, shown in Equation \ref{quantileestimate}, as an evaluation of the empirical distribution function $\hat{\mathcal{H}}_{M}(s)$ at $\hat{\vartheta}_{1}$. As the number of permutations $M$ approaches the total number of possible permutations $T!$, the empirical distribution $\hat{\mathcal{H}}_{M}(s)$ trivially converges to $\hat{\mathcal{H}}(s)$ defined in Equation \ref{fullpermutationdistribution}. 
		\vskip0in
		From the results of Lemmas \ref{lemma1} and \ref{lemma2}, we know that as $T\rightarrow\infty$, $\hat{\vartheta}_{1}\xrightarrow{D}\mathcal{F}$, where the continuous distribution $\mathcal{F}$ has expectation $\vartheta$ and variance $\varphi^{-1}\tau_\omega^2 + \left(\varphi\mathscr{R}\right)^{-1}\tau_r^2$. Under the null, each realization $\hat{\vartheta}_m$ is independently drawn from $\mathcal{F}$. From the Glivenko-Cantelli Theorem as $T\rightarrow\infty$,
		\begin{align}
			\sup_{s} \left|\hat{\mathcal{H}}(s) - \mathcal{F}(s)\right| \xrightarrow{a.s.} 0.
		\end{align}
		\vskip0in
		We define $\mathcal{U}\sim\text{Uniform}(0,1)$ and note that if $s\sim\mathcal{G}$, we can write $s\sim\mathcal{G}^{-1}(\mathcal{U})$. We state that $\mathcal{F}\left(\lim_{T\rightarrow\infty}\hat{\vartheta}_1\right)$ follows the same distribution as $\mathcal{F}\left(\mathcal{F}^{-1}\left(\mathcal{U}\right)\right)\sim\mathcal{U}$.
		\vskip0in
		Combining these results and Theorem \ref{theorem1} with the continuous mapping theorem in Equation \ref{cmt} and Slutsky's Theorem,
		\begin{align}
			\lim_{M\rightarrow T!}\hat{\mathcal{H}}_{M}\left(\hat{\vartheta}_{1}\right) &\rightarrow \hat{\mathcal{H}}\left(\hat{\vartheta}_{1}\right)\\
			\lim_{T\rightarrow\infty}\hat{\mathcal{H}}\left(\hat{\vartheta}_{1}\right) &\rightarrow \mathcal{F}\left(\hat{\vartheta}_{1}\right)\\
			\lim_{T\rightarrow\infty}\mathcal{F}\left(\hat{\vartheta}_{1}\right) &\xrightarrow{D} \mathcal{F}\left(\lim_{T\rightarrow\infty}\hat{\vartheta}_{1}\right)\label{cmt}\\
			\mathcal{F}\left(\lim_{T\rightarrow\infty}\hat{\vartheta}_{1}\right) &\sim \mathcal{U},
		\end{align}
		so under Conditions \ref{samplingcondition} - \ref{modelerrorcondition},
		\begin{align}
			\lim_{T\rightarrow\infty} \lim_{M\rightarrow T!} \hat{Q}_{M}  = \lim_{T\rightarrow\infty} \lim_{M\rightarrow T!} \hat{\mathcal{H}}_{M}\left(\hat{\vartheta}_{1}\right) \xrightarrow{D} \mathcal{U}.
		\end{align}
	\end{proof}
	\par
	The proof of Theorem \ref{theorem4} follows from the result of Theorem \ref{theorem1}.
	\begin{proof}[Proof of Theorem \ref{theorem4}]
		Under the alternative, we have $\vartheta_1<\vartheta_i$ for all possible permutations $i=1,\ldots,T!$. We observe $M$ permutations of the total possible. Define $\delta_M$ as the minimum difference between the underlying parameter $\vartheta_1$ and another permutation in $m=1,\ldots,M$, and $\eta_M$ as the maximum estimation error over all permutations $m=1,\ldots,M$.
		\begin{align}
			\delta_M &= \min_{1 < m \leq M} \left|\vartheta_m-\vartheta_1\right|\geq \delta \hskip0.1in\text{where }\delta = \lim_{M\rightarrow T!}\delta_M\\
			\eta_M &= \max_{1 \leq m \leq M} \left|\hat{\vartheta}_{m}-\vartheta_m\right|\leq \eta\hskip0.1in\text{where }\eta = \lim_{M\rightarrow T!}\eta_M
		\end{align}
		If $2\eta_M < \delta$, then the quantile estimate $\hat{Q}_M$ will return the true value in the limit. Writing the minimum using the definition in Equation \ref{quantileestimate},
		\begin{align}
			\min_{2\eta_M < \delta} \hat{Q}_M &= \min_{2\eta_M < \delta} \frac{1}{M}\sum_{m=1}^{M} \mathbf{1}\left\{\hat{\vartheta}_{m} \leq \hat{\vartheta}_{1}\right\}\\
			&= \min_{2\eta_M < \delta} \frac{1}{M}\sum_{m=1}^{M} \mathbf{1}\left\{\hat{\vartheta}_{m} - \hat{\vartheta}_{1} \leq 0\right\}\\
			&= \min_{2\eta_M < \delta} \frac{1}{M}\sum_{m=1}^{M} \mathbf{1}\left\{(\hat{\vartheta}_{m} - \vartheta_m) + (\vartheta_1 - \hat{\vartheta}_{1}) + (\vartheta_m - \vartheta_1) \leq 0\right\}\label{minline1}.
		\end{align}
		The quantity in Equation \ref{minline1} will reach its minimum when the estimation errors are maximized, leading to a larger value inside the indicator and fewer pairs that meet the criteria.
		\begin{align}
			\min_{2\eta_M < \delta} \hat{Q}_M &\geq \min_{2\eta_M < \delta}\frac{1}{M}\sum_{m=1}^{M} \mathbf{1}\left\{2\eta_M + (\vartheta_m - \vartheta_1) \leq 0\right\}\\
			&\geq \min_{2\eta_M < \delta}\frac{1}{M}\sum_{m=1}^{M} \mathbf{1}\left\{2\eta_M + \vartheta_m \leq \vartheta_1\right\}\label{minline2}
		\end{align}
		For the maximum, 
		\begin{align}
			\max_{2\eta_M < \delta} \hat{Q}_M
			&= \max_{2\eta_M < \delta} \frac{1}{M}\sum_{m=1}^{M} \mathbf{1}\left\{(\hat{\vartheta}_{m} - \vartheta_m) + (\vartheta_1 - \hat{\vartheta}_{1}) + (\vartheta_m - \vartheta_1) \leq 0\right\}\label{maxline1}\\
			&\leq \max_{2\eta_M < \delta}\frac{1}{M}\sum_{m=1}^{M} \mathbf{1}\left\{-2\eta_M + (\vartheta_m - \vartheta_1) \leq 0\right\}\\
			&\leq \max_{2\eta_M < \delta}\frac{1}{M}\sum_{m=1}^{M} \mathbf{1}\left\{\vartheta_m \leq \vartheta_1+2\eta_M\right\}\label{maxline2}.
		\end{align}
		To show consistent behavior with the original ordering, if $2\eta_M < \delta$ and $\vartheta_m\leq\vartheta_1$,
		\begin{align}
			2\eta_M + \vartheta_m  \leq 2\eta_M + \vartheta_1 - \delta &\leq \vartheta_1\\ 
			\hskip0.1in\text{and} \hskip0.1in\vartheta_m \leq \vartheta_1-\delta + 2\eta_M&\leq \vartheta_1 + 2\eta_M,
		\end{align}
		or if $2\eta_M<\delta$ and $\vartheta_m>\vartheta_1$,
		\begin{align}
			2\eta_M + \vartheta_m > 2\eta_M + \vartheta_1 + \delta &> \vartheta_1\\  
			\hskip0.1in\text{and} \hskip0.1in\vartheta_m > \vartheta_1 + \delta &> \vartheta_1 + 2\eta_M.
		\end{align}
		Thus, the ordering is preserved. 
		\begin{align}
			\lim_{M\rightarrow T!} \min_{2\eta_M < \delta} \hat{Q}_M &\geq \lim_{M\rightarrow T!} \frac{1}{M}\sum_{m=1}^{M} \mathbf{1}\left\{\vartheta_m \leq \vartheta_1\right\}\geq Q\\
			\lim_{M\rightarrow T!} \max_{2\eta_M < \delta} \hat{Q}_M &\leq \lim_{M\rightarrow T!} \frac{1}{M}\sum_{m=1}^{M} \mathbf{1}\left\{\vartheta_m \leq \vartheta_1\right\}\leq Q\\
			&\implies \lim_{M\rightarrow T!} \hat{Q}_M = Q\hskip0.1in\text{when $2\eta_M<\delta$}
		\end{align}
		We now need to show the probability of $2\eta_M<\delta$ goes to one in the limit. Pick any $\delta>0$.
		\begin{align}
			\mathbb{P}\left(\eta_M <\frac{\delta}{2}\right) &= \mathbb{P}\left(\max_{1 \leq m \leq M} \left|\hat{\vartheta}_{m}-\vartheta_m\right|<\frac{\delta}{2}\right)
		\end{align}
		From the result of Theorem \ref{theorem1}, when $\tau_\omega^2=0$,
		\begin{align}
			\lim_{\mathscr{R}\rightarrow\infty} \lim_{T\rightarrow\infty} \mathbb{P}\left(\left|\hat{\vartheta}_{m}-\vartheta_m\right|<\frac{\delta}{2}\right) &= 1.
		\end{align}
		Similarly, 
		\begin{align}
			\lim_{\varphi\rightarrow\infty} \lim_{T\rightarrow\infty} \mathbb{P}\left(\left|\hat{\vartheta}_{m}-\vartheta_m\right|<\frac{\delta}{2}\right) &= 1.
		\end{align}
		Returning to the maximum,
		\begin{align}
			\mathbb{P}\left(\max_{1 \leq m \leq M} \left|\hat{\vartheta}_{m}-\vartheta_m\right|<\frac{\delta}{2}\right) &= \mathbb{P}\left(\left|\hat{\vartheta}_{1}-\vartheta_1\right| < \frac{\delta}{2}, \ldots,\left|\hat{\vartheta}_{M}-\vartheta_{M}\right| < \frac{\delta}{2}\right)\\
			&\geq 1 - \sum_{m=1}^M\mathbb{P}\left(\left|\hat{\vartheta}_{m}-\vartheta_m\right| \geq \frac{\delta}{2}\right),
		\end{align}
		and taking the limit under the two scenarios,
		\begin{align}
			\lim_{\mathscr{R}\rightarrow\infty} \lim_{T\rightarrow\infty} \mathbb{P}\left(\max_{1 \leq m \leq M} \left|\hat{\vartheta}_{m}-\vartheta_m\right|<\frac{\delta}{2}\right) &\geq 1 - \sum_{m=1}^M\lim_{\mathscr{R}\rightarrow\infty} \lim_{T\rightarrow\infty}\mathbb{P}\left(\left|\hat{\vartheta}_{m}-\vartheta_m\right| \geq \frac{\delta}{2}\right) = 1,
		\end{align}
		when $\tau_\omega^2=0$, and
		\begin{align}
			\lim_{\varphi\rightarrow\infty} \lim_{T\rightarrow\infty} \mathbb{P}\left(\max_{1 \leq m \leq M} \left|\hat{\vartheta}_{m}-\vartheta_m\right|<\frac{\delta}{2}\right) &\geq 1 - \sum_{m=1}^M\lim_{\varphi\rightarrow\infty} \lim_{T\rightarrow\infty}\mathbb{P}\left(\left|\hat{\vartheta}_{m}-\vartheta_m\right| \geq \frac{\delta}{2}\right) = 1.
		\end{align}
		This implies $\hat{Q}_M \xrightarrow{P} Q$ in both cases presented in the statement of Theorem \ref{theorem4}.  
	\end{proof}
	\par
	The proof of Theorem \ref{theorem5} follows a step by step derivation of the distribution from the known starting point of Condition \ref{modelerrorcondition}.
	\begin{proof}[Proof of Theorem \ref{theorem5}]
		We begin with the distribution of the model residuals that is known from Equation \ref{modelerrorconditionequation} and Condition \ref{modelerrorcondition}, $\mathbf{U}_{m,\omega,r} \sim \mathcal{MN}_{T\times d}\left(\mathbf{0},\mathbf{I}_{T},\mathbf{S}_{m,\omega,r}\right)$. The predicted residuals can be expressed as a function of the model residuals, like in Equation \ref{testresiduals}, and their distribution follows from that of $\mathbf{U}_{m,\omega,r}$ as in Equation \ref{Rdistribution}. Like in the proof of Lemma \ref{lemma1}, we define 
		\begin{align}
			\bm{\Phi}_{m,\omega,r,k} = \mathbf{H}_{m,\omega,r,k}\left[\mathbf{H}_{m,\omega,r,-k}'\mathbf{H}_{m,\omega,r,-k}\right]^{-1}\mathbf{H}_{m,\omega,r,k}',
		\end{align}
		and further define 
		\begin{align}
			\mathbf{V}_{m,\omega,r,k}\sim\mathcal{MN}_{T_{k}\times d}(\mathbf{0}, \mathbf{I}_{T_{k}}, \mathbf{I}_d)
		\end{align}
		such that we can write $\mathbf{R}_{m,\omega,r,k} = \left(\bm{\Phi}_{m,\omega,r,k} + \mathbf{I}_{T_{k}}\right)^{1/2}\mathbf{V}_{m,\omega,r,k}\mathbf{S}_{m,\omega,r}^{1/2}$.
		The distribution has an equivalent vectorized form, with covariance matrix the Kronecker product between the row-wise and column-wise variation.
		\begin{align}
			\mathbf{r}_{m,\omega,r,k} = \text{vec}\left(\mathbf{R}_{m,\omega,r,k}\right) &\sim\mathcal{N}_{T_{K}d}\left(\mathbf{0}, \left[\bm{\Phi}_{m,\omega,r,k} + \mathbf{I}_{T_k}\right] \otimes \mathbf{S}_{m,\omega,r}\right)\label{vecRdistribution}
		\end{align}
		For simplicity, define $\bm{\Delta}_{m,\omega,r,k} = \left[\bm{\Phi}_{m,\omega,r,k} + \mathbf{I}_{T_k}\right] \otimes \mathbf{S}_{m,\omega,r}$ and write the equivalent form $\mathbf{r}_{m,\omega,r,k} = \bm{\Delta}_{m,\omega,r,k}^{1/2}\mathbf{v}_{m,\omega,r,k}$, where $\mathbf{v}_{m,\omega,r,k}\sim\mathcal{N}_{T_kd}\left(\mathbf{0},\mathbf{I}\right)$. The covaraince matrix $\bm{\Delta}_{m,\omega,r,k} = \left(\delta_{m,\omega,r,k,ij}\right)$ is symmetric and each element in the vector $\mathbf{v}_{m,\omega,r,k} = \left(v_{m,\omega,r,k,i}\right)\sim \mathcal{N}(0,1)$.
		\begin{align}
			\mathbf{r}_{m,\omega,r,k}'\mathbf{r}_{m,\omega,r,k} &= \mathbf{v}_{m,\omega,r,k}'\bm{\Delta}_{m,\omega,r,k}\mathbf{v}_{m,\omega,r,k}\\
			&= \sum_{i=1}^{T_kd} \bigg( \delta_{m,\omega,r,k,ii}v_{m,\omega,r,k,i}^2 + \sum_{\substack{j=1\\ j\neq i}}^{T_kd} \delta_{m,\omega,r,k,ij}v_{m,\omega,r,k,i}v_{m,\omega,r,k,j}\bigg)\\
			&= \sum_{i=1}^{T_kd} \bigg( \delta_{m,\omega,r,k,ii}v_{m,\omega,r,k,i}^2 + 2\sum_{j=1}^{i-1} \delta_{m,\omega,r,k,ij}v_{m,\omega,r,k,i}v_{m,\omega,r,k,j}\bigg)\\
			&= \sum_{i=1}^{T_kd} \bigg( \delta_{m,\omega,r,k,ii}v_{m,\omega,r,k,i}^2 + \frac{1}{2}\sum_{j=1}^{i-1} \delta_{m,\omega,r,k,ij}\left(v_{m,\omega,r,k,i}+v_{m,\omega,r,k,j}\right)^2 \nonumber\\
			&\hskip1in - \frac{1}{2}\sum_{j=1}^{i-1}\delta_{m,\omega,r,k,ij}\left(v_{m,\omega,r,k,i}-v_{m,\omega,r,k,j}\right)^2 \bigg)
		\end{align}
		\vskip0in
		We note that 
		\begin{align}
			v_{m,\omega,r,k,i}^2 &\sim \chi_1^2\\
			\frac{1}{2}(v_{m,\omega,r,k,i}+v_{m,\omega,r,k,j})^2 &\sim \chi^2_1\\
			\frac{1}{2}(v_{m,\omega,r,k,i}-v_{m,\omega,r,k,j})^2 &\sim \chi^2_1,
		\end{align}
		and $\delta_{m,\omega,r,k,ij}$ depend on the matrices $\bm{\Phi}_{m,\omega,r,k} = (\phi_{m,\omega,r,k,ij})$ and $\mathbf{S}_{m,\omega,r}=(s_{m,\omega,r,ij})$. Let $i'=\lceil i/d\rceil$,  $j'=\lceil j/d\rceil$, $i^*=i \bmod d$, and $j^*=j \bmod d$. We define the following random variables for all $\omega\in\Omega_\text{obs}$, $r=1,\ldots,\mathscr{R}$, $k=1,\ldots,K$, $i=1,\ldots,T_kd$ and $j<i$.
		\begin{align}
			X_{m,\omega,r,k,i},Y_{m,\omega,r,k,ij}, Z_{m,\omega,r,k,ij} &\sim \chi_1^2
		\end{align}
		We can write each $\delta_{m,\omega,r,k,ij} = \left(\phi_{m,\omega,r,k,i'j'} + \mathbf{1}\left\{i' = j'\right\}\right)s_{m,\omega,r,k,i^*j^*}$ and input the definitions.
		\begin{align}
			\mathbf{r}_{m,\omega,r,k}'\mathbf{r}_{m,\omega,r,k} &= \sum_{i=1}^{T_kd} \bigg[\left(\phi_{m,\omega,r,k,i'i'}+1\right)s_{m,\omega,r,k,i^*i^*}X_{m,\omega,r,k,i} \nonumber\\
			&\hskip0.5in + \sum_{j=1}^{d\lfloor (i-1) / d \rfloor} \phi_{m,\omega,r,k,i'j'}s_{m,\omega,r,k,i^*j^*}\left(Y_{m,\omega,r,k,ij} - Z_{m,\omega,r,k,ij}\right) \nonumber\\
			&\hskip0.5in + \sum_{j=d\lfloor (i-1) / d\rfloor + 1}^{i-1}\left(\phi_{m,\omega,r,k,i'j'} + 1\right)s_{m,\omega,r,k,i^*j^*}\left(Y_{m,\omega,r,j,ij} - Z_{m,\omega,r,k,ij}\right) \bigg]
		\end{align}
		To obtain the distribution for the variation parameter, we aggregate this generalized chi-square distribution over the test sets $k=1,\ldots,K$, featurizations $r=1,\ldots,\mathscr{R}$, and potential realizations $\omega\in\Omega_\text{obs}$.
	\end{proof}

\end{document}